\documentclass[a4paper]{article}
\usepackage{paper}
\pdfoutput=1

\def\constant{0.6597}

\def\baseformulaaa#1{\constant^{#1}\cdot #1!}
\def\formula#1{\lfloor\constant^{#1}(#1!)\rfloor}
\def\formulaa#1{\lfloor\constant^{#1}(#1)!\rfloor}
\def\formulaaa#1{\lfloor\constant^{#1}\cdot #1!\rfloor}

\def\cconstant{0.7838}

\def\cbaseformulaaa#1{\cconstant^{#1}\cdot #1!}
\def\cformula#1{\lfloor\cconstant^{#1}(#1!)\rfloor}
\def\cformulaa#1{\lfloor\cconstant^{#1}(#1)!\rfloor}
\def\cformulaaa#1{\lfloor\cconstant^{#1}\cdot #1!\rfloor}

\newcommand\repeattheorem[2]{
\begingroup
\def\thetheorem{#1}
\begin{theorem}
    #2
\end{theorem}
\addtocounter{theorem}{-1}
\endgroup
}

\newenvironment{proof-sketch}[1][\unskip]{
\begin{proof}[Proof sketch (full proof in appendix)]}{\end{proof}}

\title{On Lower Bounds for Maximin Share Guarantees}
\author{
    Halvard Hummel
    \vspace{0.5em}
    \\
    \normalsize
    Norwegian University of Science and Technology
    \vspace{-0.3em}
    \\
    \small
    \texttt{halvard.hummel@ntnu.no}
}
\date{}

\begin{document}

\maketitle

\begin{abstract}
    \noindent
    We study the problem of fairly allocating a set of indivisible items to a
    set of agents with additive valuations. Recently, \citeauthor{Feige:2022}
    (WINE'21) proved that a maximin share (MMS) allocation exists for all
    instances with $n$ agents and no more than $n + 5$ items. Moreover, they
    proved that an MMS allocation is not guaranteed to exist for instances with
    $3$ agents and at least $9$ items, or $n \ge 4$ agents and at least $3n + 3$
    items. In this work, we shrink the gap between these upper and lower bounds
    for guaranteed existence of MMS allocations. We prove that for any integer
    $c > 0$, there exists a number of agents $n_c$ such that an MMS allocation
    exists for any instance with $n \ge n_c$ agents and at most $n + c$ items,
    where $n_c \le \formulaaa{c}$ for allocation of goods and $n_c \le
    \cformulaaa{c}$ for chores. Furthermore, we show that for $n \neq 3$ agents,
    all instances with $n + 6$ goods have an MMS allocation.
\end{abstract}

\section{Introduction}

We are interested in the problem of fairly dividing a set of resources or tasks
to a set of agents---a problem that frequently arises in day-to-day life and has
been extensively studied since the seminal work of \citet{Steinhaus:48}. While
the classical setting assumes that the resources are infinitely
\emph{divisible}, a variant of the problem in which a set of \emph{indivisible}
items are to be allocated has been studied extensively in the last couple of
decades (see, e.g., \citet{Amanatidis:2022} and \citet{Suksompong:21} for
recent, detailed overviews).

For indivisible items, classical fairness measures, such as \emph{envy-freeness}
and \emph{proportionality}, are no longer guaranteed. Instead, relaxed
fairness measures are considered, such as the \emph{maximin share} (MMS)
\emph{guarantee} \cite{Budish:2011}. For the MMS guarantee, each agent should
receive a set of items worth at least as much as she could guarantee herself if
she were to partition the items into bundles and got to choose a bundle last.
Surprisingly, it is not guaranteed that an allocation of this kind exists
\cite{Procaccia:14}. In fact, there exists problem instances for which at least
one agent receives a bundle worth no more than $39/40$ of her MMS
\cite{Feige:2022}. However, good approximations exist and can be found
efficiently. The best current approximation algorithm guarantees each agent
at least $3/4 + 1/(12n)$ of her MMS, where $n$ is the number of
agents \cite{Garg:20c}.

When valuations are additive, MMS allocations are guaranteed to exist in certain
special cases, such as when there are at most $n + 5$ items \cite{Feige:2022} or
the set of valuation functions is restricted in certain ways
\cite{Amanatidis:17a,Heinen:18}. Our goal in this paper is to further improve
these existence results for MMS allocations---showing that the number of items
an instance can have scales with the number of agents, beyond one item per
agent.

We are interested in improving this lower bound for existence to further
determine the usefulness of MMS as a fairness measure, especially in real-world
scenarios. Usage of the online fair allocation tool Spliddit
\cite{Spliddit:23} suggests that many real-world instances have few agents
and on average a few times as many items as agents \cite{Caragiannis:19}.  As
the upper bound for existence is currently at around three times as many items
as agents \cite{Feige:2022}, reducing the gap between the two bounds betters our
understanding of these cases.

\subsection{Contributions}

In this work, we improve on the known bound for the number of goods, $m$, an
instance with $n$ agents can have and be guaranteed to have an MMS allocation.
We find that there exists some function $f(n) = \omega(\sqrt{\lg n})$ such that
an MMS allocation exists for all instances with $m \le n + f(n)$ goods,
improving on the result of $m \le n + 5$ \cite{Feige:2022}.\footnote{Expressing
$f(n)$ in terms of $n$ is nontrivial, due to the factorial in
\cref{thr:existence-stronger}.} Specifically, for any integer $c > 0$ we prove
the following bound for the required number of agents for guaranteed MMS
existence in instances with $m \le n + c$ goods.

\newcommand\existencegoodscontent{
    For any integer $c > 0$, there exists an $n_c \le \formula{c}$ such that all
    instances with $n \ge n_c$ agents and no more than $n + c$ goods have an MMS
    allocation.
}

\begin{theorem}\label{thr:existence-stronger}
    \existencegoodscontent
\end{theorem}

It has been shown by counterexample that $c = 5$ is the largest constant such
that an MMS allocation always exists for all instances with any number $n$ of
agents and at most $n + c$ goods \cite{Feige:2022}. We show that when $n \neq
3$, an MMS allocation always exists when $c = 6$.

\newcommand\csixtheoremcontent{
    For an instance with $n \neq 3$ agents, an MMS allocation always exists if
    there are $m \le n + 6$ goods.
}
\begin{theorem}\label{thr:c=6}
    \csixtheoremcontent
\end{theorem}

In a similar fashion to $c = 6$, which is shown by case analysis, we also find
that for $c = 7$ it is sufficient to have $n \ge 8$ for MMS existence.

\newcommand\cseventheoremcontent{
    For an instance with $m = n + 7$ goods, an MMS allocation always exists if
    there are $n \ge 8$ agents.
}
\begin{theorem}\label{thr:c=7}
    \cseventheoremcontent
\end{theorem}

Finally, we show that there exists a similar existence guarantee for chores as
was shown for goods in \cref{thr:existence-stronger}.

\newcommand\existencechorescontent{
    For any integer $c > 0$, there exists an $n_c \le \cformula{c}$ such that
    all instances with $n \ge n_c$ agents and no more than $n + c$ chores have
    an MMS allocation.
}

\begin{theorem}\label{thr:existence-chores}
    \existencechorescontent
\end{theorem}

Our proofs of \cref{thr:existence-stronger,thr:existence-chores} build on two
new structural properties of \emph{ordered instances}.\footnote{Instances in
which the agents have the same preference order over the items.}
First and most importantly, we exploit a common structure in MMS partitions for
ordered instances with $m \le 2n$. When an ordered instance has $n$ agents
and $m = n + c$ items for some constant $c \ge 0$, each agent has an MMS partition in
which the $n - c$ most valuable (least valuable for chores) items appear in
bundles of size one. If $c \le n$, the remaining $2c$ items must be placed in
the remaining $c$ bundles. The number of ways $2c$ items can be partitioned
into $c$ bundles depends only on $c$. Thus, as $n$ increases, more agents will
have similar MMS partitions.

To analyse the number of agents required for there to be enough similarity for
an MMS allocation to exist, we impose a partial ordering over the bundles, based
on the concept of domination.  Due to the common preference order in ordered
instances, we can for some pairs of bundles $B$ and $B'$ determine that $B$ is
better than $B'$ no matter the valuation function. In this case, we say that $B$
\emph{dominates} $B'$. A trivial example is when $B$ and $B'$ differ by only a
single item.  When a sufficient number of agents have bundles in their MMS
partitions that form a chain in the domination based partial ordering, a
reduction to a smaller instance can be found. By employing induction, we use an
upper bound for the size of the maximum antichain to obtain the existence
bounds.

\subsection{Related Work}

The existence of MMS has been the focus of a range of publications in recent
years.  Early experiments failed to yield problem instances for which no MMS
allocation exists \cite{Bouveret:16}. \citet{Procaccia:14} later found a way to
construct counterexamples for any number of agents $n \ge 3$.\footnote{For $n <
3$, MMS allocations always exist.} These counterexamples used a number of goods
that was exponential in the number of agents. The number of goods needed for a
counterexample was later reduced to $3n + 4$ by \citet{Kurokawa:16} and recently
to $3n + 3$ by \citet{Feige:2022}.\footnote{$3n + 1$ when $n$ is even.} In the
opposite direction, \citet{Bouveret:16} showed that all instances with at most
$n + 3$ goods have MMS allocations, later improved to $n + 4$ by
\citet{Kurokawa:16} and $n + 5$ by \citet{Feige:2022}. \citeauthor{Feige:2022}
also found an instance with $3$ agents and $9$ goods for which no MMS allocation
exists.

While MMS allocations do not always exist, it has been shown that
they exist with high probability, under certain simple
assumptions~\cite{Kurokawa:16,Suksompong:16,Amanatidis:17a}.

The existence of MMS allocations has also been explored in cases where valuation
functions are restricted. \citet{Amanatidis:17a} showed that when item values
are restricted to the set $\{0, 1, 2\}$, an MMS allocation always exists. Later,
\citet{Heinen:18} studied existence for Borda and lexicographical valuation
functions.

There is also a rich literature on finding approximate MMS allocations, either
by providing each agent with a bundle worth at least a percentage of her
MMS \cite{Amanatidis:17a,Garg:18a,Gourves:19,Garg:20c,Ghodsi:21,Feige:22} or
providing a percentage of the agents with bundles worth at least MMS
\cite{Hosseini:21}.

While the main focus of the literature has been on goods, some work has been
done on MMS for chores, both on existence \cite{Aziz:17c,Feige:2022} and
approximation \cite{Aziz:17c,Barman:20e,Huang:21,Feige:22}.

\section{Preliminaries}\label{sec:prelims}

An instance $I = \langle N, M, V \rangle$ of the \emph{fair allocation problem}
consists of a set $N = \{1, 2, \dots, n\}$ of \emph{agents} and a set $M = \{1,
2, \dots, m\}$ of \emph{items}. Additionally, there is a collection $V$ of $n$
\emph{valuation functions}, $v_i : 2^M \rightarrow \mathbb{R}$, one for each
agent $i \in N$. To simplify notation, we let both $v_{ij}$ and $v_{i}(j)$
denote $v_i(\{j\})$ for $j \in M$. We assume that the valuation functions are
additive, i.e., $v_i(M) = \smash{\sum_{g \in M} v_i(g)}$, with $v_i(\emptyset) =
0$. We deal, seperately, with two types of items: \emph{goods}, which have
non-negative value, $v_i(j) \ge 0$, and \emph{chores}, which have non-positive
value, $v_i(j) \le 0$.\footnote{By this definition, an item $j$ with $v_{ij} =
0$ is both a good and a chore. However, as we do not consider mixed instances,
the overlapping definitions do not matter.} \emph{Mixed instances}, which
consist of a mix of goods and chores, and perhaps have items that are goods for
some agents and chores for others, will not be considered. Hence, the valuation
functions are monotone, i.e., for $S \subseteq T \subseteq M$, $v_i(S) \le
v_i(T)$ for goods and $v_i(S) \ge v_i(T)$ for chores. For simplicity, we assume
throughout the paper that all instances consist of goods, except in
\cref{sec:chores}, which covers instances consisting of only chores.

For any instance $I = \langle N, M, V \rangle$, we wish to partition the items
in $M$ into $n$ \emph{bundles}, one for each agent. An $n$-partition of $M$ is
called an \emph{allocation}. We are interested in finding allocations that
satisfy the fairness measure known as the \emph{maximin share guarantee}
\cite{Budish:2011}. That is, we wish to find an allocation in which each agent
gets a bundle valued at no less than what she would get if she were to partition
the items into bundles and got to choose her own bundle last.

\begin{definition}
    For an instance $I = \langle N, M, V \rangle$, the \emph{maximin share
    (MMS)} of an agent $i \in N$ is given by
    \[
        \mu_i^I = \max_{A \in \Pi_I}\min_{A_j \in A} v_i(A_j),
    \]
    \noindent where $\Pi_I$ is the set of all possible allocations in $I$. If
    obvious from context, the instance is omitted, and we write simply $\mu_i$.
\end{definition}

We say that an allocation $A = \langle A_1, A_2, \dots,
A_n\rangle$ \emph{satisfies the MMS guarantee} or, simply, is an \emph{MMS
allocation}, if each agent $i \in N$ receives a bundle valued at no less than
her MMS, i.e., $v_i(A_i) \ge \mu_i$. For a given agent $i \in N$ we call any
allocation $A$ in which $v_i(A_j) \ge \mu_i$ for every bundle $A_j \in A$, an
\emph{MMS partition} of $i$ for $I$. By definition, each agent has at least one
MMS partition for any instance $I$, but can possibly have several.

Several useful properties of MMS have been discovered in previous work. Perhaps
the most useful, is the concept of \emph{ordered instances}, in which the agents
have the same \emph{preference order} over the items.

\begin{definition}
    Instance $I = \langle N, M, V \rangle$ is said to be \emph{ordered} if
    $v_{ij} \ge v_{i(j + 1)}$ for all $i \in N$ and $1 \le j <
    |M|$.
\end{definition}

\citeauthor{Bouveret:16} showed that both for existence and approximation
results, it is sufficient to consider only ordered instances.

\begin{lemma}[\citeauthor{Bouveret:16}, \citeyear{Bouveret:16}]
    For any instance $I = \langle N, M, V \rangle$, there exists an ordered
    instance $I'$, with $\mu_i^I = \mu_i^{I'}$ for all $i \in N$, and for any
    allocation $A'$ for $I'$ there exists an allocation $A$ for $I$ such that
    $v_i(A_i) \ge v'_i(A'_i)$ for all $i \in N$.
\end{lemma}

The instance $I'$ is constructed by sorting the item valuations of each agent
and reassigning them to the items in a predetermined order. The MMS of an agent
does not change from $I$ to $I'$, due to the inherent one-to-one map between
items in $I$ and $I'$. Allocation $A$ can be constructed from $A'$ by going
through the items in order from most to least valuable, letting the agent $i$
that received item $j$ in $A'$ select her most preferred remaining item in $I$.
Since there are at least $j$ items in $I$ with an equivalent or greater value
than $j$ has in $I'$, at least one of these must remain when $i$ selects an item
for $j$ and the selected item has at least as high value in $I$ as $j$ has in
$I'$. Consequently, each agent's bundle in $A$ is at least as valuable as in
$A'$.

Another useful form of instance simplification, that we will rely heavily on, is
the concept of \emph{valid reductions}. A valid reduction is, simply put, an
allocation of a subset of the items to a subset of the agents, where each agent
receives a satisfactory bundle,\footnote{A bundle $B$ is satisfactory for an agent
$i$ if $v_i(B) \ge \mu_i$, or in the case of approximation $v_i(B) \ge
\alpha\mu_i$ for some $\alpha > 0$.} while the MMS of the remaining agents is not
smaller in the new, smaller instance.

\begin{definition}
    Let $I = \langle N, M, V \rangle$ be an instance. Removing a subset of items
    $M' \subseteq M$ and a subset of agents $N' \subseteq N$ is called a
    \emph{valid reduction} if there exists a way to allocate the items in
    $M'$ to the agents in $N'$ such that each agent $i' \in N'$ receives a
    bundle $B_{i'}$ with $v_{i'}(B_{i'}) \ge \mu_{i'}^I$ and for $i \in N
    \setminus \{N'\}$, we have $\mu_i^{I'} \ge \mu_i^{I}$, where $I' = \langle N
    \setminus N', M \setminus M', V' \rangle$.
\end{definition}

Valid reductions are commonly used when finding approximate MMS allocations,
where several simple reductions have been found~%
\cite{Kurokawa:16,Amanatidis:17a,Ghodsi:18,Garg:18a,Garg:20c}.
These reductions allocate a small number of goods to a single agent---providing
a powerful tool when considering instances with only a few more goods than
agents. Most of these reductions can also be used in the existence case and the
ones relevant to us are given below. Proofs for their validity can be found in
the papers cited above. For completness we also prove them in the appendix,
along with other omitted proofs.

\begin{lemma}\label{lem:single-item-reduction}
    Let $I = \langle N, M, V \rangle$ be an instance. If there is agent $i \in
    N$ and good $j \in M$ with $v_{ij} \ge \mu_i$, then allocating $\{j\}$ to
    $i$ is a valid reduction.
\end{lemma}

\begin{lemma}\label{lem:2-item-reduction}
    Let $I = \langle N, M, V \rangle$ be an instance. If there is an agent $i
    \in N$ and distinct goods $j, j' \in M$ with $v_i(\{j, j'\}) \ge \mu_i$ and
    $v_{i'}(\{j, j'\}) \le \mu_{i'}$ for all $i' \in (N \setminus \{i\})$, then
    allocating $\{j, j'\}$ to $i$ is a valid reduction.
\end{lemma}

\begin{lemma}\label{lem:pigeonhole-reduction-2}
    Let $I = \langle N, M, V \rangle$ be an ordered instance. If there is an
    agent $i \in N$ with $v_i(\{n, n + 1\}) \ge \mu_i$, then allocating $\{n, n
    + 1\}$ to $i$ is a valid reduction.
\end{lemma}

\begin{lemma}\label{lem:2-item-reduction-from-1}
    Let $I = \langle N, M, V \rangle$ be an ordered instance. If there is an
    agent $i \in N$ and good $j \in M$ such that $v_{i}(j) \ge \mu_i$ and
    $v_{i'}(j) < \mu_{i'}$ for all $i' \in N \setminus \{i\}$, then allocating
    $\{j, j'\}$ to $i$, where $j'$ is the worst good in $M \setminus \{j\}$, is a
    valid reduction.
\end{lemma}

In addition to valid reductions, there are several cases in which an MMS
allocation is known to exist. These cases will be used as base cases in our
existence argument.

\begin{lemma}\label{lem:identical-valuation-mms}
    Let $I = \langle N, M, V \rangle$ be an instance. If there are at least $n -
    1$ agents with the same MMS partition, then an MMS allocation exists.
\end{lemma}

\begin{lemma}\label{lem:mms-existence-few-agents}
    An MMS allocation always exists for an instance $I = \langle N, M, V
    \rangle$, where $n \le 2$.
\end{lemma}

\begin{lemma}[\citeauthor{Feige:2022}, \citeyear{Feige:2022}]\label{lem:existance-feige}
    An MMS allocation always exists for an instance $I = \langle N, M, V
    \rangle$ if $m \le n + 5$.
\end{lemma}

\section{Existence For Any Constant}\label{sec:existence-any-constant}

Our first main result is that for any $c > 0$, there exists an $n_c > 0$ such
that all instances with $n \ge n_c$ agents and $n + c$ goods have MMS
allocations. To show this, we exploit a structural similarity in MMS partitions
when $c < n$. Specifically, if $m < 2n$, any MMS partition contains some bundles
of cardinality zero or one.\footnote{If there is a bundle of cardinality zero in
an MMS partition of agent $i$, then $\mu_i = 0$.} For ordered instances of this
kind, there is a set of at least $n - c$ goods valued, individually, at MMS or
higher by each agent, namely the set of the $n - c$ most valuable goods:

\begin{lemma}\label{lem:goods-at-mms}
    Let $I = \langle N, M, V \rangle$ be an ordered instance with $m = n + c$
    for some $c$ with $n > c > 0$. Then $v_{ij} \ge \mu_i$ for all $i \in N$ and
    $j \in \{1, 2, \dots, n - c\}$.
\end{lemma}

\begin{proof}
    Agent $i \in N$ either has $\mu_i = 0$ or each bundle in any one of her MMS
    partitions contains at least one good. If $\mu_i = 0$, then $v_{ij} \ge
    \mu_i$ for all $j \in M$. Otherwise, at most $c$ of the bundles in an MMS
    partition can contain more than one good. The worst good $g$ contained in a
    bundle of cardinality one, is such that $g \ge n - c$. Since $\mu_i \le
    v_{ig}$ by definition, $\mu_i \le v_{ig} \le v_i(n - c) \le v_i(n - c - 1)
    \le \dots \le v_i(1)$.
\end{proof}

The shared set of goods valued at MMS or higher guarantees that each agent has
an MMS partition where these goods appear in bundles of cardinality one.

\begin{lemma}\label{lem:mms-partition-layout}
    Given an ordered instance $I = \langle N, M, V \rangle$ and agent $i \in N$,
    let $k$ denote the number of goods in $M$ valued at $\mu_i$ or higher by
    $i$. Then $i$ has an MMS partition in which each of the goods $1, 2,
    \dots, \min(n - 1, k)$ forms a bundle of cardinality one.
\end{lemma}

\begin{proof}
    Let $A$ be an arbitrary MMS partition of $i$, $B_g \in A$ denote the bundle
    containing some $g \in M$ and let $G_A = \{g \in \{1, 2, \dots, \min(n - 1,
    k)\} : |B_g| > 1\}$. If $G_A = \emptyset$, then all the goods $1, 2, \dots,
    \min(n - 1, k)$ appear in bundles of cardinality one. We wish to show that
    if $G_A \neq \emptyset$, then there exists an MMS partition $A'$ with
    $|G_{A'}| < |G_A|$. Assume that $G_A \neq \emptyset$ and for some $g \in
    G_A$, select $A_j \in A$ such that $\{1, 2, \dots, \min(n - 1, k)\} \cap A_j
    = \emptyset$. Then, the allocation $A' = \langle A_1, \dots, \{g\}, \dots,
    A_j \cup (B_g \setminus \{g\}), \dots, A_n \rangle$ is an MMS partition of
    $i$, as $v_i(A_j \cup (A_g \setminus \{g\})) \ge v_i(A_j) \ge \mu_i$ and
    $v_{ig} \ge \mu_i$. Further, as only the two bundles $B_g$ and $A_j$ have
    been modified, and $A_j$ did not contain any good in $\{1, 2, \dots, \min(n
    - 1, k)\}$, we have $|G_{A'}| = |G_{A}| - 1$. Hence, $i$ has an MMS
    partition $A^*$ with $G_{A^*} = \emptyset$.
\end{proof}

\Cref{lem:mms-partition-layout} enforces a particularly useful restriction on
the set of $n$-partitions of $M$ when $n > c$. As a result of
\cref{lem:goods-at-mms}, \cref{lem:mms-partition-layout} guarantees that each
agent has at least one MMS partition in which the $n - c$ most valuable goods
appear in bundles of cardinality one. In this MMS partition, the remaining
$2c$ goods are partitioned into $c$ bundles. Ignoring the possibility of having
empty bundles, the number of ways to partition these $2c$ goods into $c$ bundles
is $\smash{{ 2c \brace c }}$, where $\smash{{ 2c \brace c }}$ is a Stirling
number of the second kind.\footnote{If there is an empty bundle, then all
$n$-partitions, including those without empty bundles, are MMS partitions of the
agent.}

The value of $\smash{{ 2c \brace c }}$ does not depend on the
value of $n$.  Thus, as the number of agents increases, there must eventually be
multiple agents with the same MMS partition. Specifically, when there are ${2c
\brace c}(c - 2) + 1$ agents, at least $c - 1$ of them share the same MMS
partition of the type outlined in \cref{lem:mms-partition-layout}. Then an MMS
allocation can be constructed by allocating the goods $1, 2, \dots, n - c$ to $n
- c$ of the other $n - c + 1$ agents. The last of the $n - c + 1$ agents
receives her favorite remaining bundle in the shared MMS partition, and the
last $c - 1$ agents each receives an arbitrary remaining bundle in the shared
MMS partition. This is an MMS allocation, as all but one agent receives a bundle
from one of her MMS partitions, and the remaining agent $i$ receives a bundle
worth at least $(v_i(M) - v_i(\{1, 2, \dots, n - c\}))/c \ge (c\mu_i)/c =
\mu_i$.

While the above argument is sufficient for showing existence for any $c > 0$,
the lower bounds of \citet{Rennie:69} on Stirling numbers give $n_c = {2c \brace
c }(c - 2) + 1 > c^c$. Hence, while straightforward, the argument is not
sufficient to prove the bound of \cref{thr:existence-stronger}, $n_c \le
\formulaaa{c}$. For that, we will use a more involved inductive argument.

Our inductive procedure builds on the observation that a full MMS allocation
need not be found directly. Instead, for a $c > 0$, it is sufficient to
use valid reductions to reduce to some smaller instance with a smaller number
$c' \ge 0$ of additional goods. As long as the smaller instance has at least
$n_{c'}$ agents, an MMS allocation exists for the original instance. Here, the
existence for $n' \ge n_{c'}$ with $m' \le n' + c'$ is assumed to be proven,
with \cref{lem:existance-feige,thr:c=6} as base cases. To show the existence of
valid reductions, we will again exploit the structure of the MMS partitions
guaranteed by \cref{lem:goods-at-mms,lem:mms-partition-layout} in order to
construct an upper bound on the number of agents required before some agents
have MMS partitions with additional shared structure.

To construct valid reductions, and as a definition of shared structure, we will
utilize a partial ordering of bundles. For ordered instances, it is often
possible to say that some subset of goods $B \subseteq M$ is at least as good as
some other subset $B' \subseteq M$, no matter the valuation function.
Obviously, this holds when $B' \subseteq B$, even for non-ordered instances.
However, due to the common preference-order of the agents, it could be that $B$
is better than $B'$ even when $B' \not\subseteq B$. For example, when $B = \{3,
7, 8, 11, 14\}$ and $B' = \{6, 7, 11, 13\}$. As illustrated in
\cref{fig:domination-example}, $B$ is at least as valuable as $B'$, since
$v_i(3) \ge v_i(6)$, $v_i(8) \ge v_{i}(13)$, $\{7, 11\} \subset B$, and $\{7,
11\} \subset B'$. We can formalize the partial ordering in the following way.

\begin{definition}\label{def:domination}
    For an ordered instance $I = \langle N, M, V \rangle$, a subset of goods $B
    \subseteq M$ \emph{dominates} a subset of goods $B' \subseteq M$ if there is
    an injective function $f : B' \rightarrow B$ such that $f(j) \le j$ for all
    $j \in B'$. If $B$ dominates $B'$, we denote this by $B \succeq B'$. We use
    $B \succ B'$ for the case where $B \neq B'$.
\end{definition}

The domination ordering provides a useful set of valid reductions. Whenever an
agent $i$ values a bundle $B$ at MMS or higher, and every other agent in the
instance has a bundle in her MMS partition that dominates $B$, then allocating
$B$ to $i$ forms a valid reduction.

\begin{lemma}\label{lem:domination-reduction}
    Let $I = \langle N, M, V \rangle$ be an ordered instance and $B$ be a bundle
    with $v_i(B) \ge \mu_i$ for some $i \in N$. If each agent $i' \in N
    \setminus \{i\}$ has a bundle $B_{i'}$ in her MMS partition with $B_{i'}
    \succeq B$, then allocating $B$ to $i$ is a valid reduction.
\end{lemma}

\begin{proof}
    For any agent $i' \in N \setminus \{i\}$, we wish to show that her MMS is at
    least as high in the reduced instance as in the original instance. Since
    $B_{i'} \succeq B$, there exists an injective function $f_{i'} : B
    \rightarrow B_{i'}$ with $f_{i'}(g) \le g$ for $g \in B$. We will show that an
    MMS partition of $i'$ can be turned into a $n$-partition containing $B$ and
    $n - 1$ bundles valued at $\mu_{i'}$ or higher. Then, in the reduced
    instance, the MMS of $i'$ cannot be less than the value of the
    least valuable bundle among these $n - 1$ bundles, which has a value of
    at least $\mu_{i'}$. The conversion is done by performing the following
    steps on an MMS partition of $i'$ containing $B_{i'}$.
    \begin{enumerate}
        \item Go through the goods $g \in B$ from least to most valuable,
            exchanging the position of $g$ and $f_{i'}(g)$ in the
            partition.\label{step:domination-reduction-1}
        \item Move all goods in $B_{i'} \setminus B$ to any other bundle in the
            partition.\label{step:domination-reduction-2}
    \end{enumerate}
    \noindent Since $f_{i'}(g) \le g$, after exchanging the position of $g$ and
    $f_{i'}(g)$ in step \ref{step:domination-reduction-1}, $g$ will not move.
    Further, since $f_{i'}$ is injective, $f_{i'}(g)$ will not be moved before
    it is exchanged with $g$. Thus, since $f_{i'}(g) \in B_{i'}$ before the
    step, $B \subseteq B_{i'}$ after all the exchanges.  Additionally, after
    step \ref{step:domination-reduction-1} the value of any other bundle in the
    partition cannot have decreased, as $v_{i'}(g) \le v_{i'}(f_{i'}(g))$. As
    adding an item to a bundle does not decrease the value of the bundle, step
    \ref{step:domination-reduction-2} does not decrease the value of other
    bundles than $B_{i'}$. Thus, afterwards, $B_{i'} = B$ and the value
    of each other bundle remains at least $\mu_{i'}$.
\end{proof}

\begin{figure}
    \centering
    \begin{tikzpicture}[scale=0.5]
        \draw[step=1.0,black,thin] (0, 0) grid (14, 1);
        \draw[step=1.0,black,thin,yshift=0.5cm] (0, 1) grid (14, 2);
        \foreach \i in {1,...,14} {
            \node at (\i - 0.5, 2.9) {\footnotesize \i};
        }

        \foreach \i/\j in {3/6,7/7,8/11,11/13}{
            \draw (\j-0.5,0.5) edge[thick,-stealth](\i-0.5,2);
        }

        \node[inner sep=0pt] at (-1, 2) {$B\phantom{'}$};
        \node[inner sep=0pt] at (-1, 0.5) {$B'$};

        \begin{pgfonlayer}{bg}
            \foreach \j in {3,7,8,11,14} {
                \draw[fill=gray!75] (\j - 1, 1.5) rectangle (\j, 2.5);
            }
            \foreach \j in {6,7,11,13} {
                \draw[fill=gray!75] (\j - 1, 0) rectangle (\j, 1);
            }
        \end{pgfonlayer}
    \end{tikzpicture}
    \caption{A bundle $B$ dominating a bundle $B'$ in an instance with $14$
    goods. The arrows represent a possible function $f$ (out of the two possible
    functions).}
    \label{fig:domination-example}
\end{figure}
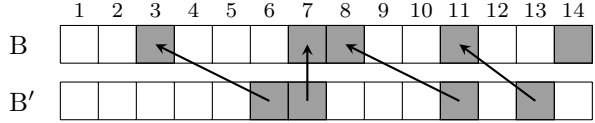

In order to find valid reductions through the domination ordering, we will
consider bundles that are of the same size $k \ge 2$.\footnote{Bundles of size 1
immediately induce a valid reduction.} When
two bundles of size $k$ share a subset of $k - 1$ goods, we know that one
dominates the other, as each bundle only contains one good in addition to the
shared subset.\footnote{The bundles may be equal. However, by definition they
dominate each other when equal.} With multiple bundles of size $k$ that all
share the same subset of $k - 1$ goods, at least one of the bundles is dominated
by all the other bundles. Thus, if for some $(k - 1)$-sized subset of goods $S
\subset M$, each agent has a $k$-sized bundle containing $S$ in one of her MMS
partitions, then there exists a valid reduction that removes one agent and $k$
goods.

\begin{lemma}\label{lem:k-sized-domination-reduction}
    Let $I = \langle N, M, V \rangle$ be an ordered instance, $k > 0$ an
    integer, and $S \subset M$ a subset of $k - 1$ goods. For each agent $i
    \in N$, let $B_i$ be a bundle in an MMS partition of $i$ such that $|B_i| =
    k$ and $S \subset B$. Then, there is an agent $i' \in N$ such that
    allocating $B_{i'}$ to $i'$ is a valid reduction.
\end{lemma}

\begin{proof}
    Let $g = \max\{g' : i \in N, g' \in (B_i \setminus S)\}$. Then, for any $i
    \in N$, $B_i \succeq (S \cup \{g\})$ and there is $i' \in N$ such that $B_{i'}
    = S \cup \{g\}$. By \cref{lem:domination-reduction}, giving $B_{i'}$ to $i'$
    is a valid reduction.
\end{proof}

Making use of \cref{lem:k-sized-domination-reduction} requires an instance where
all agents share similar $k$-sized bundles in one of their MMS partitions---a
property that usually does not hold for arbitrary instances. However, for any
integer $c > 0$, \cref{lem:goods-at-mms,lem:mms-partition-layout} guarantee that
when $n > c$, all agents have MMS partitions in which any bundle of size greater
than one is a subset of the $2c$ worst goods.
Thus, as the number of agents increases, there will eventually be some $k \ge 2$
for which some set $S$ of $k - 1$ goods is shared between $k$-sized bundles in
the MMS partitions of multiple agents. When there are at least $c$ such agents,
the combination of \cref{lem:goods-at-mms,lem:k-sized-domination-reduction}
provides a way to create a valid reduction removing $n' \le n - c + 1$ agents
and $n' + k - 1$ goods. Simply allocate one of $1, 2, \dots, n - c$ to each of
the at most $n - c$ agents without an MMS partition containing a $k$-sized
bundle with subset $S$, and use the method of
\cref{lem:k-sized-domination-reduction} to allocate a $k$-sized bundle to one of
the remaining agents. This approach can be used in our inductive argument as
long as $n - n' \ge n_{c - k + 1}$. In other words, there must be at least
$\max(c, n_{c - k + 1} + 1)$ agents with a $k$-sized bundle in one of their MMS
partitions that has $S$ as a subset.

To obtain the bound in \cref{thr:existence-stronger}, we will, instead of using
$\max(c, n_{c - k + 1} + 1)$, show that if $n_{c'} \ge n_{c' - 1} + 1$ for $c' >
6$ and there are at least $c$ agents with a $k$-sized bundle in their MMS
partition, then we only need $\max(c - k + 1, n_{c - k + 1} + 1)$ agents with a
$k$-sized bundle sharing the same $(k - 1)$-sized subset of goods.  Our proof
relies on a result of \citet{Aigner-Horev:22} on envy-free
matchings.\footnote{Their result has previously been used in MMS approximation.}
In this setting, for a graph $G$ and set $X$ of vertices in $G$, $N_G(X)$
denotes the union of the open neighbourhood in $G$ of each vertex in $X$.

\begin{definition}
    A matching $M$ in a bipartite graph $G = (X \cup Y, E)$ is \emph{envy-free
    with regards to $X$} if no unmatched vertex in $X$ is adjacent in $G$ to a
    matched vertex in $Y$.
\end{definition}

\begin{theorem}[\citeauthor{Aigner-Horev:22}, \citeyear{Aigner-Horev:22}]\label{thr:envy-free-matching}
    Given a bipartite graph $G = (X \cup Y, E)$, there exists a non-empty
    envy-free matching with regards to $X$ if $|N_G(X)| \ge |X| \ge 1$.
\end{theorem}

Using \cref{thr:envy-free-matching} and the assumptions described above, we show
that an MMS allocation exists if an agent has an MMS partition with at most one
bundle containing more than two goods.

\begin{lemma}\label{lem:one-3+-good-bundle}
    Let $I = \langle N, M, V \rangle$ be an ordered instance, with $m = n + c$
    goods for some $c > 0$ and assume that for $c > c' > 5$, there exists an
    integer $n_{c'} > 0$ such that all instances with $n' \ge n_{c'}$ agents and
    $m' = n' + c'$ goods have MMS allocations and $n_{c'} > n_{c' - 1}$ for $c'
    > 6$. Then, if $n > n_{c - 1}$ and an agent $i \in N$ has an MMS partition
    $A$ with at least $n - 1$ bundles of size less than three, an MMS allocation
    exists.
\end{lemma}

\begin{proof-sketch}
    If $\mu_i = 0$, the result follows from \cref{lem:pigeonhole-reduction-2}. If
    $\mu_i > 0$, each bundle in the MMS partition, except for potentially one of
    size three or more, has size one or two. We wish to show that unless there
    exists a perfect matching of agents to bundles they value at MMS or more in
    the MMS partition, there instead exists a non-empty envy-free matching that
    only contains bundles of size one or two.  Given such an envy-free matching,
    a valid reduction that removes $x$ agents and $2x$ goods can be found by
    allocating all bundles of size two in the matching before applying
    \cref{lem:2-item-reduction-from-1} to each bundle of size one.

    To find a non-empty envy-free matching, we exploit that Hall's marriage
    theorem allows us to create a subgraph with fewer agents than bundles, where
    an envy-free matching in the subgraph is envy-free in the original graph.
    Agent $i$ will be present in the subgraph, as bundles are from $i$'s MMS
    partition. Furthermore, there are fewer agents than bundles in the subgraph.
    Thus, we can additionally remove the bundle of size three or larger, unless
    already removed, while \cref{thr:envy-free-matching} still guarantees a
    non-empty envy-free matching.
\end{proof-sketch}

Using \cref{lem:one-3+-good-bundle} we can improve our lower bound on number of
goods valued at MMS or higher by an agent $i$ based on the size of the bundles
in their MMS partitions.

\begin{lemma}\label{lem:goods-at-MMS}
    Let $I = \langle N, M, V \rangle$ be an ordered instance, with $m = n + c$
    goods for some $c > 0$ and assume that for any $c' > 5$, there exists an
    integer $n_{c'} > 0$ such that all instances with $n' \ge n_{c'}$ agents and
    $m' = n' + c'$ goods have MMS allocations and $n_{c'} > n_{c' - 1}$ for $c'
    > 6$. If $n > n_{c - 1}$ and agent $i \in N$ has an MMS partition $A$ with a
    bundle of size $k > 2$, then either $v_i(n - c + k - 1) \ge \mu_i$ or an MMS
    allocation exists.
\end{lemma}

\begin{proof}
    If $\mu_i = 0$, then $v_i(n - c + k - 1) \ge \mu_i$. Now, assume that $\mu_i
    > 0$, and as a result $k \le c + 1$. If $v_i(n - c + k - 1) < \mu_i$, then
    at most $n - c + k - 2$ bundles in $A$ have size one and no bundle is empty.
    Of the remaining bundles, there is one of size $k$ and the $c - k + 1$
    others contain at least two goods each. These bundles of size at least two,
    contain the remaining $n + c - (n - c + k - 2) - k = 2(c - k + 1)$ goods.
    Thus, each of these bundles contains exactly two goods, and $A$ contains a
    single bundle of cardinality greater than $2$. Consequently, an MMS
    allocation exists by \cref{lem:one-3+-good-bundle}.
\end{proof}

As is evident from \cref{lem:goods-at-MMS}, if there are $c$ agents with a
$k$-sized bundle in their MMS partition, we can give $k - 1$ of them a bundle
worth MMS or higher by allocating them each one of the goods $n - c + 1, n - c +
2, \dots, n - c + k - 1$. For our domination-based reduction, as long as there
are $c$ agents in the instance with a $k$-sized bundle in their MMS partition,
then we only need $\max(c - k + 1, n_{c - k + 1} + 1)$ agents with $k$-sized
bundles with a shared $(k - 1)$-sized subset. We are now ready to prove
\cref{thr:existence-stronger}.

\repeattheorem{\ref{thr:existence-stronger}}{\existencegoodscontent}

\begin{proof}
    For $c \le 5$, \cref{lem:existance-feige} guarantees that an MMS allocation
    always exists for any number of agents. Further, \cref{thr:c=6}, which is
    proven without \cref{thr:existence-stronger}, guarantees that an MMS
    allocation always exists when $c = 6$ and $n \ge 4 < \formulaaa{6}$. Thus,
    there exists an $n_c \le \formulaaa{c}$ for $c < 7$ and we only need to
    consider cases where $c \ge 7$.

    We wish to show that for every integer $c \ge 7$, all instances with $n \ge
    \formulaaa{c}$ agents and $m \le n + c$ goods have an MMS allocation. To
    obtain this result, we will use induction with $c < 7$ as base case.
    For a given value of $c \ge 7$, assume that for every integer $c'$ with $6 <
    c' < c$, an MMS allocation exists when there are $n' \ge n_{c'} =
    \formulaaa{c'}$ agents and at most $n' + c'$ goods. Note that under this
    assumption we know that $\formulaa{c' - 1} < \formulaaa{c'}$ for all values
    of $c'$. Hence, we are able to use the results of
    \cref{lem:one-3+-good-bundle,lem:goods-at-MMS} and only show existence for
    instances with $n \ge \formulaaa{c}$ and $m = n + c$.

    Let $I = \langle N, M, V\rangle$ be an ordered instance of $n$ agents and $m
    = n + c$ goods, where $n \ge \formula{c}$. We will show that under the
    inductive assumption, $I$ has an MMS allocation. Let $A_I(i)$ be an MMS
    partition of agent $i \in N$ of the type described by
    \cref{lem:mms-partition-layout}, maximizing the number of bundles of
    cardinality one. To show that $I$ has an MMS allocation, we will consider
    domination between particularly bad bundles in $A_I(i)$ of different agents.
    Let $B_I(i)$ be a bundle in $A_I(i)$ in which the best good $g$ is such that
    $n \le g$. Observe that $B_I(i) \subseteq \{n, n + 1, \dots, n + c\}$. Thus,
    if $|B_I(i)| = k$ for some integer $k$, $B_I(i)$ is one of ${c + 1 \choose
    k}$ possible $k$-sized subsets of $\{n, n + 1, \dots, n + c\}$.

    Before proceeding, we will deal with some special cases, to simplify and
    tighten the further analysis. If for any agent $i \in N$ it holds that
    $\mu_i = 0$ or $|B_I(i)| \le 2$, then $v_i(\{n, n - 1\}) \ge \mu_i$ and an
    MMS allocation exists by \cref{lem:pigeonhole-reduction-2}. If $|B_I(i)| > c
    - 1$, then an MMS allocation exists by \cref{lem:one-3+-good-bundle}.
    Furthermore, if $\mu_i > 0$ and $|B_I(i)| = c - 1$, then either an MMS
    allocation exists by \cref{lem:one-3+-good-bundle} or $A_I(i)$ contains $n -
    2$ bundles of size one and $v_i(n - 2) \ge \mu_i$. If $v_i(n - 2) \ge
    \mu_i$, there could exist a subset $N' \subset N$ of $n - 2$ agents such
    that removing $N'$ and $\{1, 2, \dots, n - 2\}$ forms a valid reduction.
    Otherwise, there is a non-empty subset $N'' \subset N$ of agents and an
    equally-sized subset $M'' \subset M$ of at most $c$ goods such that no agent
    in $N \setminus N''$ values any good in $M''$ at MMS or higher and there
    exists a perfect matching between the agents in $N''$ and goods they value
    at MMS or higher in $M''$. The method from \cref{lem:one-3+-good-bundle} can
    be used to extend the perfect matching to a valid reduction with $|N''|$
    agents and $2|N''|$ goods. Thus, an MMS allocation exists if $|B_I(i)| = c -
    1$ for any $i \in N$.

    We can now assume that $2 < |B_I(i)| < c - 1$ and $\mu_i > 0$ for all $i \in
    N$. We wish to determine the number of agents required such that for at
    least one $k \in \{3, 4, \dots, c - 2\}$, there must be at least $\max(c -
    k + 1, n_{c - k + 1} + 1)$ agents with $|B_I(i)| = k$ and where the bundles
    $B_I(i)$ share a $(k - 1)$-sized subset of goods. Since $B_I(i) \subset \{n,
    n + 1, \dots, n + c\}$ and any bundle of size $k$ contains $k$ subsets of
    size $k - 1$, if there for a $k \in \{3, 4, \dots, c - 2\}$ is at least
    \begin{equation}\label{eq:require-one-k}
        1 + \left(\frac{{c \choose k - 1}}{k} + \frac{{c \choose k - 2}}{k -
        1}\right)\max(c - k, n_{c - k + 1})
    \end{equation}
    agents for which $|B_I(i)| = k$, then there are at least $\max(c - k + 1,
    n_{c - k + 1} + 1)$ bundles in $(B_I(i) : i \in N, |B_I(i)| = k)$ that share
    the same $(k - 1)$-sized subset of goods.\footnote{The parenthesized term in
    the equation is the number of distinct ($k - 1$)-sized subsets of $\{n, n +
    1, \dots, n + c\}$, divided by the number of distinct ($k - 1)$-sized
    subsets of a single $k$-sized bundle and separated by if they contain good
    $n$ or not.} Combining \cref{eq:require-one-k} for all possible $k$, we get
    that when there are
    \begin{equation}\label{eq:required-all-k}
        1 + \sum_{k = 3}^{c - 2}\left(\frac{{c \choose k - 1}}{k} + \frac{{c
        \choose k - 2}}{k - 1}\right)\max(c - k, n_{c - k + 1})
    \end{equation}
    agents, there is some $3 \le k \le c - 2$ for which there are at least
    $\max(c - k + 1, n_{c - k + 1} + 1)$ agents with $|B_I(i)| = k$, where the
    $B_I(i)$ share the same $(k - 1)$-sized subset.

    We wish to show that \cref{eq:required-all-k} is bounded from above by
    $\formula{c}$. In order to prove the bound, we make the following
    observations. Since $n_{c'} = \formulaaa{c'}$ for $c > c' > 0$, when $k < c
    - 2$ we have $\max(c - k, n_{c - k + 1}) = n_{c - k + 1}$. Also, since $c
    \ge 7$ we can use \cref{lem:existance-feige} to show that:
    \begin{equation}\label{eq:small-c-prime}
    \begin{aligned}
        2 &+ \sum_{k = c - 4}^{c - 2}\left(\frac{{c \choose k - 1}}{k} + \frac{{c
        \choose k - 2}}{k - 1}\right)\max(c - k, n_{c - k + 1}) \\
        &< \sum_{k = c -
        4}^{c - 2}\left(\frac{{c \choose k - 1}}{k} + \frac{{c
        \choose k - 2}}{k - 1}\right)\formulaa{c - k + 1}
    \end{aligned}
    \end{equation}
    For any $k \in \{3, 4, \dots, c - 2\}$, it holds that
    \[
        \left(\frac{{c \choose k - 1}}{k} + \frac{{c \choose k - 2}}{k -
        1}\right)\formulaa{c - k + 1} \ge c
    \]
    Combining the observations with \cref{eq:required-all-k}, we get that if
    there are at least
    \begin{equation}\label{eq:required-agents}
        -1 + \sum_{k = 3}^{c - 2}\left(\frac{{c \choose k - 1}}{k} + \frac{{c \choose
        k - 2}}{k - 1}\right)\formulaa{c - k + 1}
    \end{equation}
    agents in $I$, an MMS allocation must exist, since there is some $k \in \{3,
    4, \dots, c - 2\}$ for which there are $c$ or more agents with $|B_I(i)| =
    k$ and at least $\max(c - k + 1, n_{c - k + 1} + 1)$ of them have the same
    $(k - 1)$-sized subset of $B_I(i)$.  Thus, we must show that
    \cref{eq:required-agents} is less than or equal to our bound
    $\formulaaa{c}$. We have that for $\alpha > 0$:

    \begin{equation}\label{eq:first}
    \begin{aligned}
        \sum_{k = 3}^{c - 2}&\left(\frac{{c \choose k - 1}}{k} + \frac{{c
        \choose k - 2}}{k - 1} \right)\lfloor \alpha^{c - k + 1} (c - k +
        1)!\rfloor \\
        &\le \alpha^c \cdot c!\sum_{k = 3}^{c - 2}\left(\frac{\alpha^{-k + 1}}{k!} +
        \frac{\alpha^{-k + 1}}{k!(c - k + 1)}\right)
    \end{aligned}
    \end{equation}
    Using the Maclaurin series $e^y = \sum_{j = 0}^\infty y^j/(j!)$, we get that
    \begin{align}
        \sum_{k = 3}^{c - 2}&\left(\frac{\alpha^{-k +1}}{k!} +
        \frac{\alpha^{-k + 1}}{k!(c - k + 1)}\right) \\
        &\le \frac{1}{12 \cdot \alpha^2} + (\alpha + \frac{1}{4})\sum_{k =
        3}^{\infty} \frac{\alpha^{-k}}{k!} \\
        &= \frac{1}{12 \cdot \alpha^2} + (\alpha +
        \frac{1}{4})\left( e^{\frac{1}{\alpha}} - \frac{1}{2\cdot \alpha^2} -
        \frac{1}{\alpha} - 1 \right)\label{eq:second}
    \end{align}
    \Cref{eq:second} is equal to $1$ if $\alpha = 0.65964118 \dots$, and less
    than $1$ if $\alpha$ is larger.\footnote{The exact value $\alpha$ for which
    \cref{eq:second} is equal to $1$ can be used in
    \cref{thr:existence-stronger} instead of the rounded value $\constant$.}
    Thus, as a result of \cref{eq:first}, we know from \cref{eq:required-agents}
    that the number of required agents is less than or equal to $-1 +
    \baseformulaaa{c} < \formulaaa{c}$. Consequently, $I$ has an MMS allocation
    by our inductive hypothesis.
\end{proof}

\section{Improved Bounds for Small Constants}\label{sec:special-cases}

In the previous section, we saw that for any integer constant $c > 0$, there
exists a, rather large, number $n_c$ such that all instances with $n \ge n_c$
agents and no more than $n + c$ goods have MMS allocations. There exists some
slack in the calculations of the limit, especially for smaller values of $c$.
For example, the constant $2$ in \cref{eq:small-c-prime} used to mitigate the
floor function and $1\;+$ part in \cref{eq:required-all-k}, can be replaced by a
somewhat larger constant.  Moreover, while hard to make use of in the general
case, there exist additional, unused properties and interactions between the MMS
partitions of different agents. As a result, it is possible to, on a
case-by-case basis, show better bounds for small constants by analyzing the
possible structures of MMS partitions and their interactions for specific values
of $c$. We state the two following results for $c = 6$ and $c = 7$. Both proofs
rely on an exhaustive analysis of possible MMS partition structure combinations,
and are given in the appendix.

\repeattheorem{\ref{thr:c=6}}{\csixtheoremcontent}

\repeattheorem{\ref{thr:c=7}}{\cseventheoremcontent}

\section{Fair Allocation of Chores}\label{sec:chores}

So far we have only considered instances in which the items are \emph{goods}. In
this section, we show that a similar result to the one for goods in
\cref{thr:existence-stronger} exists for \emph{chores}. The resulting
bounds for $n_c$ are somewhat worse for chores due to minor differences in the
way that valid reductions can be constructed. The main difference is the lack of
a result equivalent to \cref{lem:pigeonhole-reduction-2}. In practice, this
means that while we for goods could ignore bundles of cardinality two in our
domination-based counting argument, we must include bundles of cardinality two
for chores. Fortunately, it is possible to show that the bundles of cardinality
two that are of interest to us are all the same bundle. Thus, the number of
agents with a bundle of cardinality two required to find a reduction is
relatively small.

To simplify notation and make the proofs for chores similar to those for goods,
we use a slightly different definition for ordered instances in this section.
The only difference is that the numbering of the items changes---while item $1$
was the best good, it is now the worst chore. In other words, we wish to
maintain the same order of absolute value for the items.

\begin{definition}
    Instance $I = \langle N, M, V \rangle$ is said to be \emph{ordered} if
    $v_{ij} \le v_{i(j + 1)}$ for all $i \in N$ and $1 \le j < |M|$.
\end{definition}

We now show that if an agent has a bundle of cardinality two in her MMS
partition, she also has a similar---both in structure and distribution of
chores---MMS partition in which the bundle of cardinality two is $\{n, n +
1\}$.

\begin{lemma}\label{lem:size-2-chores}
    Let $I = \langle N, M, V \rangle$ be an ordered instance for chores, and $i
    \in N$ an agent with an MMS partition $A$ containing a bundle $B$ with $|B|
    = 2$, $B \cap \{1, 2, \dots, n - 1\} = \emptyset$. Then $i$ has an MMS
    partition $A'$ such that (i) $|A_j| = |A'_j|$ for all $j \in N$, (ii) $\{n,
    n + 1\} \in A'$, and (iii) the position of the chores $1, 2, \dots, n - 1$
    is the same in $A$ and $A'$.
\end{lemma}

\begin{proof}
    Assume that $B = \{x, y\}$, where $x < y$. Let $A'$ be the allocation
    equivalent to $A$, except for that $x$ and $y$ have changed place with,
    respectively, $n$ and $n + 1$. We wish to show that $A'$ is an MMS partition
    and satisfies (i), (ii) and (iii). In any MMS partition, there must be at
    least one bundle $B'$ with $|B' \cap \{1, 2, \dots, n + 1\}| \ge 2$.
    Thus, $v_i(\{n, n + 1\}) \ge v_i(B') \ge \mu_i$. Since $n \le x$ and
    $n + 1 \le y$, the bundles that contained $n$ and $n + 1$ are no worse after
    the swap and $A'$ is an MMS partition of $i$.

    Since the only difference between $A$ and $A'$ is two swaps of chores, and
    $\{n, n + 1, x, y\} \cap \{1, 2, \dots, n - 1\} = \emptyset$, both (i) and
    (iii) hold. Furthermore, after the swap $B = \{n, n + 1\}$ and $B \in A'$,
    thus (ii) holds.
\end{proof}

As a result of \cref{lem:size-2-chores}, if there are at least $n_{c - 1} + 1$
agents with bundles of cardinality two in their MMS partition, then we can use
$\{n, n + 1\}$ to construct a valid reduction to an instance with $n_{c - 1}$
agents and $m = n_{c - 1} + (c - 1)$ items.

To prove \cref{thr:existence-chores}, we will now develop a similar strategy as
used for goods. The strategy for chores makes use of the domination property,
which transfers perfectly, to construct valid reductions. There is, however, one
major difference. For chores, a bundle is worse if it dominates another bundle,
rather than better. Thus, we now wish to find an agent with a bundle that
dominates bundles of many other agents. We get the following variant of
\cref{lem:domination-reduction}, proven in the exact same way.

\begin{lemma}
    Let $I = \langle N, M, V \rangle$ be an ordered instance and $B$ a bundle
    with $v_i(B) \ge \mu_i$ for some $i \in N$. If each agent $i' \in N
    \setminus \{i\}$ has a bundle $B_{i'}$ in her MMS partition with $B \succeq
    B_{i'}$, then allocating $B$ to $i$ a valid reduction.
\end{lemma}

\cref{lem:k-sized-domination-reduction} holds also for chores (with the word
\emph{goods} exchanged for \emph{chores}), by modifying the proof such that $g$
is selected using $\min$ instead of $\max$ and thus $(S \cup \{g\}) \succeq B_i$
for each $i \in N$.

For chores there exists the following, standard property on the value of each
individual chore.

\begin{lemma}\label{lem:chores-min-value}
    Let $I = \langle N, M, V \rangle$ be an ordered instance, then $v_{ig} \ge
    \mu_i$ for each $g \in M, i \in N$.
\end{lemma}

\begin{proof}
    For an agent $i \in N$ and $g \in M$, each MMS partition of $i$ has bundle
    $B$ with $g \in B$. Thus, $v_{ig} \ge v_i(B) \ge \mu_i$.
\end{proof}

Valid reductions are harder to construct for chores than for goods. Of
\cref{lem:single-item-reduction,lem:2-item-reduction,lem:pigeonhole-reduction-2,lem:2-item-reduction-from-1},
only \cref{lem:2-item-reduction} holds for chores. However, as a result of
\cref{lem:k-sized-domination-reduction} we know that if there is chore $g$ that
appears in a bundle of size $1$ in the MMS partition of at least $n - 1$ of the
agents, then there is a valid reduction consisting of $g$ and the last agent.

To prove \cref{thr:existence-chores} we start by showing that each agent has an
MMS partition of a similar structure to the one given by
\cref{lem:mms-partition-layout} for goods.

\begin{lemma}\label{lem:structure-chores}
    Given an ordered instance $I = \langle N, M, V\rangle$ and agent $i \in N$,
    let $k$ denote the maximum number of bundles of cardinality one in any MMS
    partition of $i$. Then, $i$ has an MMS partition in which the chores $1, 2,
    \dots, \max(n - 1, k)$ appear in bundles of cardinality one.
\end{lemma}

\begin{proof}
    Let $A$ be an MMS partition of $i$ that contains $k$ bundles of cardinality
    one, let $B_g \in A$ denote the bundle containing some $g \in M$ and let
    $G_A = \{g \in \{1, 2, \dots, \min(n - 1, k)\} : |B_g| > 1\}$. If $G_A =
    \emptyset$, then chores $1, 2, \dots, \min(n - 1, k)$ appear in bundles of
    cardinality one. We wish to show that if $G_A \neq \emptyset$, then there
    exist an MMS partition $A'$ with at least $\min(n - 1, k)$ bundles of
    cardinality $1$ and $|G_{A'}| < |G_A|$. Assume that $G_A \neq \emptyset$ and
    for some $g \in G_A$, select $A_j \in A$ such that $|A_j| = 1$ and $\{1, 2,
    \dots, \min(n - 1, k)\} \cap A_j = \emptyset$.  Then, moving the chore in
    $A_j$ to $B_g$ and placing $g$ in $A_j$ produces an allocation $A'$ for
    which $|G_{A'}| < |G_A|$ and $A'$ contains $k$ bundles of cardinality $1$.
    By \cref{lem:chores-min-value}, and since $v_{ig} \le v_{ig'}$, both
    modified bundles are still worth at least MMS to $i$. Hence, there exists an
    MMS partition $A^*$ of $i$ with $G_{A^*} = \emptyset$.
\end{proof}

Similarly to for goods, we can for chores use the size of the largest bundle in
an MMS partition of an agent to find a lower bound on the maximum number of
bundles of cardinality one in MMS partitions of the agent.

\begin{lemma}\label{lem:number-of-single-bundles}
    Let $I = \langle N, M, V\rangle$ be an ordered instance with $n$ agents and
    $m = n + c$ chores, where $n > c > 0$. If an agent $i$ has an MMS partition
    $A$ with a bundle of size $k \ge 2$, then $i$ has an MMS partition $A'$
    that contains at least $n - (c - k + 2)$ bundles of cardinality one.
\end{lemma}

\begin{proof}
    If $A$ does not contain at least $n - (c - k + 2)$ bundles of cardinality
    $1$, then we will show that there is a way to transform $A$ into an MMS
    partition containing at least $n - (c - k + 2)$ bundles of cardinality one.
    Let $\mathcal{B}$ be a set of $c - k + 2$ bundles in $A$, such that
    $\mathcal{B}$ contains a bundle of size $k$ and the bundles in $\mathcal{B}$
    contain at least $2(c - k + 1) + k$ chores in total. A set $\mathcal{B}$
    containing at least $2(c - k + 1) + k$ chores must exist, as if
    $\mathcal{B}$ did not contain at least $2(c - k + 1) + k$ chores, there is a
    bundle of cardinality $0$ or $1$ in $\mathcal{B}$. Further, there must then
    exist a bundle of cardinality at least $2$ in $A$, but not in $\mathcal{B}$,
    since the $m' \ge n + c - 2(c - k + 1) - k + 1 = n - (c - k + 2) + 1$ chores
    not in $B$ are distributed into $n - (c - k + 2)$ bundles. Swapping the
    bundle of cardinality at least $2$ for the bundle of cardinality $0$ or $1$
    increases the number of chores in $\mathcal{B}$, a process that can be
    repeated until $\mathcal{B}$ contains sufficiently many chores. Let $M' = M
    \setminus \cup_{B \in \mathcal{B}} B$. $M'$ is the set of chores not in $B$.
    If $|M'| < n - (c - k + 2)$, extend $M'$ by removing one and one chore from
    a bundle in $\mathcal{B}$ and adding it to $M'$ until $|M'| = n - (c - k +
    2)$. Note that since $\mathcal{B}$ contains at least $2(c - k +1) + k$
    chores, $|M'|$ cannot contain more than $n - (c - k + 2)$ chores initially.
    Hence, we now have that $|M'| = n - (c - k + 2)$.

    Since $|M'| = n - (c - k + 2)$ and there are $n - (c - k + 2)$ bundles not
    in $\mathcal{B}$, we can obtain a $n - (c - k + 2)$-partition of $M'$ by
    placing each chore in $M'$ into a separate empty bundle. The partition can
    be extended to an $n$-partition of $M$ by adding the bundles from
    $\mathcal{B}$. The $n - (c - k + 2)$ bundles created from $M'$ all have
    cardinality one and are by \cref{lem:chores-min-value} worth at least MMS to
    $i$. The remaining bundles either appeared in $A$ or are bundles in $A$ that
    have had some chores removed. In either case, each bundle is worth at least
    MMS to $i$ and the $n$-partition is an MMS partition of $i$.
\end{proof}

As for goods, if most chores appear in the same bundle in an MMS partition of an
agent, then we can say something about the existence of an MMS allocation. In
particular, we get the following result, which can later be combined with
\cref{lem:number-of-single-bundles} to ignore bundles containing at least $c$
chores.

\begin{lemma}\label{lem:mostly-<2-chores}
    Let $I = \langle N, M, V\rangle$ be an ordered instance with $m = n + c$
    chores for some $c > 0$ and assume that for $c' = c - 1$, there exist an
    integer $n_{c'} > 0$ such that all instances with $n' \ge n_{c'}$ agents and
    $m' = n' + c'$ chores have MMS allocations, where $n > n_{c' - 1}$.  Then,
    if an agent $i \in N$ has an MMS partition $A$ with at least $n - 2$ bundles
    of size less than two and at least $n - 1$ bundles of size less than three,
    an MMS allocation exists.
\end{lemma}

\begin{proof}
    If $A$ contains $n - 1$ bundles of size less than two, then allocating one
    of the bundles to each of the agents in $N \setminus \{i\}$ and giving the
    last bundle to $i$ is an MMS allocation by \cref{lem:chores-min-value}.
    Otherwise, there are $n - 2$ bundles of size less than 2, a bundle $B$ of
    size 2 and a bundle $B'$ of size at least 2. If $v_{i'}(B) \ge \mu_{i'}$ for
    some $i' \in N \setminus \{i\}$, then an MMS allocation can be found by
    allocating $B$ to $i'$, $B'$ to $i$ and the remaining bundles of size less
    than 2 to the remaining agents. If $v_{i'}(B) < \mu_{i'}$ for all $i' \in N
    \setminus \{i\}$, then allocating $B$ to $i$ is a valid reduction to an
    instance with $(n - 1) \ge n_{c'}$ agents and $n + c - 2 = (n - 1) + (c -
    1)$ chores, for which an MMS allocation exits.
\end{proof}

We are now ready to prove \cref{thr:existence-chores}. The proof proceeds in an
almost equivalent manner to the one for \cref{thr:existence-stronger}.

\repeattheorem{\ref{thr:existence-chores}}{\existencechorescontent}

\begin{proof}
    For $c \le 5$, \citet{Feige:2022} showed that an MMS allocation always
    exists for any number of agents. Thus, for $c < 6$, $n_c \le
    \cformulaaa{c}$ and we only need to consider cases where $c \ge 6$.

    We wish to show that for every integer $c \ge 7$, all instances with $n \ge
    \cformulaaa{c}$ agents and $m \le n + c$ chores have an MMS allocation. To
    obtain this result, we will use induction with $c < 6$ as base case. For
    a given value of $c$, assume that for every integer $c'$ with $5 < c' < c$,
    an MMS allocation exists when there are $n' \ge n_{c'} = \cformulaaa{c}$
    agents and at most $n' + c'$ chores. Note that under this assumption
    $\cformulaaa{(c' - 1)} < \cformula{c'}$, when $c \ge c' \ge 6$. Hence, we are
    able to use \cref{lem:mostly-<2-chores}, and only show existence for
    instances where $n \ge \cformulaaa{c}$ and $m = n + c$.

    Let $I = \langle N, M, V \rangle$ be an ordered instance of $n$ agents and
    $m = n + c$ chores, where $n > \cformula{c}$. We wish to show that under the
    inductive assymption, $I$ has an MMS allocation. Let $A_I(i)$ be an MMS
    partition of agent $i \in N$ of the type described by
    \cref{lem:structure-chores,lem:size-2-chores}, where the number of bundles
    of cardinality $1$ is maximized. To show that $I$ has an MMS allocation, we
    will consider domination between bundles in $A_I(i)$, for different agents,
    that only contain particularly good chores. Let $B_I(i)$ be a bundle in
    $A_I(i)$ in which the worst chore $g$ is such that $n \le g$.

    Before proceeding, we will deal with some special cases, to simplify and
    tighten the further analysis. If $|B_I(i)| > c - 1$ for some agent $i \in
    N$, then $A_I(i)$ contains by \cref{lem:number-of-single-bundles} at least
    $n - 2$ bundles of cardinality $1$. Since $|B_I(i)| > c - 1$, the last
    bundle has cardinality at most $2$. Thus, if $|B_I(i)| > c - 1$, an MMS
    allocation exists by \cref{lem:mostly-<2-chores}. Note that it is impossible
    that $|B_I(i)| = 1$, as since $c > 0$, $A_I(i)$ contains at most $n - 1$
    bundles of cardinality $1$, each containing a chore in $\{1, 2, \dots, n -
    1\}$. However, $\{1, 2, \dots, n - 1\} \cap B_I(i) = \emptyset$ and
    $|B_I(i)| \neq 1$. Also note that if $|B_I(i)| = 2$, then we can w.l.o.g.\
    assume that $B_I(i) = \{n, n + 1\}$, due to \cref{lem:size-2-chores}.

    Let $k \in \{3, \dots, c - 1\}$. If there are $\max(c - k + 2, n_{c - k +
    1} + 1)$ agents $i \in N$ for which $|B_I(i)| = k$ and there is a ($k -
    1$)-sized subset $S \subseteq \{n, n + 1, \dots n + c\}$ such that $S
    \subset B_I(i)$ for all off them, then a reduction to an instance with
    $\max(c - k + 1, n_{c - k + 1})$ agents and $\max(c - k + 1, n_{c - k + 1})+
    (c - k + 1)$ chores can be constructed; There is a bundle $B$, that is one
    of the $B_I(i)$ with $|B_I(i)| = k$ and $S \subset B_I(i)$, that dominates
    all the other ones. Since \cref{lem:number-of-single-bundles} guarantees
    that $A_I(i)$ contains at least $n - (c - k + 2)$ bundles of cardinality one
    when $|B_I(i)| = k$, the reduction can be constructed by giving one chore
    from $\{1, 2, \dots, n - \max(c - k + 2, n_{c - k + 1} + 1)\}$ to each of
    the $n - \max(c - k + 2, n_{c - k + 1} + 1)$ other agents and then $B$ to
    any agent that values it at MMS or higher. Each agent removed, received a
    bundle valued at MMS or higher. Further, as the MMS partition $A_I(i)$ of
    any remaining agent $i$ contains all the bundles of cardinality one given
    away, along with another bundle $B_I(i)$ that is dominated by $B$, the MMS
    of $i$ cannot have decreased. Thus, we have a valid reduction.

    Similarly, if $k = 2$, since all the agents for which $|B_I(i)| = 2$ have
    $B_I(i) = \{n, n + 1\}$, if $\max(c - k + 2, n_{c - k + 1} + 1)$ agents have
    $|B_I(i)| = 2$, then a valid reduction can be constructed in the same way.

    Thus, if there for some $k \in \{3, \dots, c - 1\}$ is at least
    \[
        1 + \left(\frac{{c \choose k - 1}}{k} + \frac{{c \choose k - 2}}{k -
        1}\right)\max(c - k  + 1, n_{c - k + 1})
    \]
    agents with $|B_I(i)| = k$, then an MMS allocation exists. Or, if there are
    $\max(c, n_{c - 1} + 1)$ agents with $|B_I(i)| = 2$. Hence, when there are
    \begin{equation}\label{eq:agents-chores}
    \begin{aligned}
        1 &+ \sum_{k = 3}^{c - 1}\left(\frac{{c \choose k - 1}}{k} + \frac{{c
          \choose k - 2}}{k - 1}\right)\max(c - k  + 1, n_{c - k + 1}) \\
          &+ \max(c - 1, n_{c - 1})
    \end{aligned}
    \end{equation}
    agents, then a valid reduction can be performed for some $k \in \{2, 3,
    \dots, c - 1\}$ to an instance for which an MMS allocation exists. We wish
    to show that \cref{eq:agents-chores} is bounded from above by
    $\cformula{c}$.

    Since $n_{c'} = \cformulaaa{c'}$ for $c > c' > 0$, when $k < c - 2$, we have
    $\max(c - k + 1, n_{c - k + 1}) = n_{c - k + 1}$. Also, since $c \ge 6$,
    \begin{equation}\label{eq:floor-trick-chores}
    \begin{aligned}
        2 &+ \sum_{k = c - 3}^{c - 1}\left(\frac{{c \choose k - 1}}{k} + \frac{{c
        \choose k - 2}}{k - 1}\right)\max(c - k + 1, n_{c - k + 1}) \\
        &< \sum_{k = c -
        3}^{c - 1}\left(\frac{{c \choose k - 1}}{k} + \frac{{c
        \choose k - 2}}{k - 1}\right)\cformulaa{c - k + 1}
    \end{aligned}
    \end{equation}
    Combining \cref{eq:agents-chores,eq:floor-trick-chores}, we get that if
    there are at least
    \begin{equation}\label{eq:chores-agents}
    \begin{aligned}
        -1 &+ \sum_{k = 3}^{c - 1}\left(\frac{{c \choose k - 1}}{k} + \frac{{c
        \choose k - 2}}{k - 1}\right)\cformulaa{c - k + 1} \\
        &+ \cformulaa{c - 1}
    \end{aligned}
    \end{equation}
    agents, then an MMS allocation exists. Thus, we must show that this number
    of agents is at most $\cformulaaa{c}$. We have that for $\alpha > 0$
    \begin{equation}\label{eq:first-chores}
    \begin{aligned}
        \lfloor\alpha^{c - 1}&(c - 1)!\rfloor + \sum_{k = 3}^{c -
        1}\left(\frac{{c \choose k - 1}}{k} + \frac{{c \choose k - 2}}{k -
        1}\right)\lfloor \alpha^{c - k + 1}(c + k - 1)!\rfloor  \\
        &\le \alpha^c\cdot c!\left(\frac{1}{\alpha c} + \sum_{k = 3}^{c -
        1}\left(\frac{\alpha^{-k + 1}}{k!} + \frac{\alpha^{-k + 1}}{k!(c - k +
          1)}\right) \right)
    \end{aligned}
    \end{equation}
    Using the Macluarin series $e^y = \sum_{j = 0}^{\infty} y^j/(j!)$, and
    assuming $\alpha \ge 0.6$, we get that
    \begin{align}
        \frac{1}{\alpha c} &+ \sum_{k = 3}^{c - 1}\left(\frac{\alpha^{-k +1}}{k!} +
        \frac{\alpha^{-k + 1}}{k!(c - k + 1)}\right) \\
        &\le \frac{1}{10 \cdot \alpha^2} + (\alpha + \frac{1}{2})\sum_{k =
        3}^{\infty} \frac{\alpha^{-k}}{k!} + \left(\frac{1}{6\alpha} -
          \sum_{k=6}^{\infty} \frac{\alpha^{-k}}{k!}\right)\\
        &=
        \begin{aligned}
            &\frac{1}{6\alpha} + \frac{1}{10 \cdot \alpha^2} +
            \frac{1}{6\alpha^3} + \frac{1}{24\alpha^4} + \frac{1}{120\alpha^5}\\
            &+(\alpha - \frac{1}{2})\left( e^{\frac{1}{\alpha}} -
            \frac{1}{2\cdot \alpha^2} - \frac{1}{\alpha} - 1
            \right)\label{eq:second-chores}
        \end{aligned}
    \end{align}
    \Cref{eq:second-chores} is equal to $1$ if $\alpha = 0.78370709\dots$, and
    less than $1$ if $\alpha$ is larger.\footnote{The exact value $\alpha$ for
    which \cref{eq:second-chores} is equal to $1$ can be used in
    \cref{thr:existence-chores} instead of the rounded value $\cconstant$.}
    Thus, as a result of \cref{eq:first-chores}, we know that the number of
    required agents from \cref{eq:chores-agents} is less than or equal to $-1 +
    \cbaseformulaaa{c} < \cformulaaa{c}$. Consequently, $I$ has an MMS
    allocation by our inductive hypothesis.
\end{proof}

\section{Conclusion and Future Work}\label{sec:discussion}

\Cref{thr:existence-stronger,thr:existence-chores} show that instances with $n$
agents and $n + c$ items will for any $c > 0$ have an MMS allocation if $n$ is
sufficiently large. The required value for $n$ does, however, grow exponentially
in $c$. As a consequence, even if an instance contains the entire human
population, the value of $c$ can at most be $15$ for goods and $14$ for chores.
Thus, the result is mostly of use for instances with few agents, such as the
motivating real-world instances, where the value for $c$ is comparably large.
For these instances, \cref{thr:c=6,thr:c=7} also play a crucial role, as their
constants are relatively large in relation to the small number of
required agents.

We only know that an MMS allocation is \emph{not} guaranteed to exist when there
are about three times as many items as agents. It would be interesting to
further reduce the gap between the known upper and lower bounds.  While we have
not been able to improve our lower bounds, we find it probable that the bounds
can be improved by an approach that builds upon our domination-based partial
ordering. By better understanding how quickly large chains must appear in the
ordering, one can potentially replace \cref{eq:require-one-k} and
\cref{eq:required-all-k} by smaller terms and obtain a better bound.

There are several ways in which smaller terms could be found. First, the
argument used for \cref{thr:existence-stronger} makes use of only a single
bundle of each agent when counting the number of agents required before a
domination-based reduction must exist. Considering multiple bundles for each
agent could perhaps result in an improved bound. To this end, it is possible to
show that there is an MMS partition of each agent where, in addition to the
structure imposed by \cref{lem:mms-partition-layout}, no bundle dominates
another bundle except when both bundles have cardinality one. Second, while we
only consider domination when two bundles of size $k$ share a ($k - 1$)-sized
subset of items, a bundle of size $k$ may dominate another bundle of size $k$
even if there is no such shared subset. For example, the bundle $\{n + 2, n
+5\}$ dominates the bundle $\{n + 4, n + 6\}$. As \cref{eq:require-one-k}
considers all ($k - 1$)-sized subsets, domination interactions without shared
subsets will exist between some bundles before there are $n_c$ agents. The
difficulty in also considering such interactions is properly quantifying the
required number of agents. Furthermore, there may also be potential in
constructing reductions based on domination between bundles of different
cardinality, such as $k$ and ($k - 1$)-sized bundles, given that one is able to
properly quantify the combined number of required bundles of these sizes.

\bibliography{paper}

\clearpage
\appendix

\section*{\huge Appendix}

\section{Missing Proofs From \cref{sec:prelims}}

In this section, we present proofs for
\cref{lem:single-item-reduction,lem:2-item-reduction,lem:pigeonhole-reduction-2,%
lem:2-item-reduction-from-1,lem:identical-valuation-mms,lem:mms-existence-few-agents}.
Proofs of these lemmas can also be found in the works cited in the paragraph
preceding the lemmas in \cref{sec:prelims}. However, for some of the lemmas,
\cref{lem:domination-reduction} can be used to produce shorter, simpler proofs.

\begin{proof}[Proof of \cref{lem:single-item-reduction}] For any $i' \in N
    \setminus \{i\}$ and any MMS partition $A$ of $i'$, there is a bundle
    $B$ containing $j$. If the instance was ordered, bundle $B$ would dominate
    $\{j\}$. While domination-based reductions do not
    generally work for unordered instances, when the bundle given
    away is a subset of a bundle in the MMS partition of each agent, then the
    same argument as in \cref{lem:domination-reduction} can be used.
\end{proof}

\begin{proof}[Proof of \cref{lem:2-item-reduction}]
    For any $i' \in N \setminus \{i\}$ there are two possibilities. Either $j$
    and $j'$ appear in the same bundle in some MMS partition of $i'$ or
    the two goods appear in distinct bundles, $B$ and $B'$. In the first case,
    the bundle dominates $\{j, j'\}$ and \cref{lem:domination-reduction} can be
    used.  In the latter case, since $v_{i'}(\{j, j'\}) \le \mu_{i'}$ we have
    that $v_{i'}((B \cup B') \setminus \{j, j'\}) = v_{i'}(B \cup B') -
    v_{i'}(\{j, j'\}) \ge 2\mu_{i'} - v_{i'}(\{j, j'\}) \ge \mu_{i'}$ and an ($n
    - 1$)-partition of $M \setminus \{j, j'\}$ in which all bundles are worth at
    least $\mu_{i'}$ to $i'$ can be obtained by merging $B$ and $B'$ before
    removing $j$ and $j'$.
\end{proof}

\begin{proof}[Proof of \cref{lem:pigeonhole-reduction-2}]
    For any $i' \in N \setminus \{i\}$, every MMS partition $A$ of $i'$
    contains a bundle $B$ where $|B \cap \{1, 2, \dots, n + 1\}| \ge 2$. We have
    that $B \succeq \{n, n + 1\}$ and this is a valid reduction by
    \cref{lem:domination-reduction}.
\end{proof}

\begin{proof}[Proof of \cref{lem:2-item-reduction-from-1}]
    For any $i' \in N \setminus \{i\}$, since $v_{i'}(j) < \mu_{i'}$, in any MMS
    partition $A$ of $i'$, $j$ appears in a bundle $B$ with $|B| > 1$. Since $j
    \in B$ and $j'$ is the worst good in $M$, $B \succeq \{j, j'\}$ and this is
    valid reduction by \cref{lem:domination-reduction}.
\end{proof}

\begin{proof}[Proof of \cref{lem:identical-valuation-mms}]
    If all agents share an identical MMS partition $A$, then $A$ is an MMS
    allocation. Otherwise, let $A$ be the shared MMS partition and $i$ the agent
    for which $A$ is not an MMS partition. Since $v_i(M) \ge n\mu_i$, at least
    one bundle $A_j \in A$ is worth no less than $\mu_i$ to $i$. An MMS
    allocation can be constructed by allocating $A_j$ to $i$ and the other
    bundles in $A$ to the other $n - 1$ agents.
\end{proof}

\begin{proof}[Proof of \cref{lem:mms-existence-few-agents}]
    Follows directly from \cref{lem:identical-valuation-mms}.
\end{proof}

\section{Proof of \cref{thr:c=6}}

In our proof of \cref{thr:c=6} and later in the proof of \cref{thr:c=7}, we will
categorize partitions of the goods into types based on the cardinality of the
bundles in the partition. We say that an $n$-partition is of type $(a_1, a_2,
\dots, a_n)$ with $a_i \in \{0, 1, 2, \dots, m\}$ and $a_i \le a_{i + 1}$ if the
$a_i$ are the cardinalities of the bundles in the partition. For example, the
$2$-partition $\langle \{1, 2\}, \{3\} \rangle$ is of type $(1, 2)$.

To prove \cref{thr:c=6} we will make use of \cref{lem:feige-congregation}, a
slightly generalized aggregation of the results in Propositions 20--28 in the
full version of \citeauthor{Feige:2022}'s paper~%
\cite{Feige:2022}.\footnote{Available at https://arxiv.org/abs/2104.04977.}

\begin{lemma}\label{lem:feige-congregation}
    Let $I = \langle N = \{1, 2, 3\}, M = \{1, \dots, 9\}, V \rangle$ be an
    ordered instance and for each $i \in N$ let $x_i$ be such that $0 \le x_i
    \le \mu_i$. There always exists an allocation $A = \langle A_1, A_2, A_3
    \rangle$ with $v_i(A_i) \ge x_i$ for all $i \in N$, unless there for each $i
    \in N$ is only a single way to partition $M$ into three bundles $B_{i1}$,
    $B_{i2}$ and $B_{i3}$, such that $v_i(B_{ij}) \ge x_i$ for all $j \in N$,
    and where for two agents $|B_{i1}| = |B_{i2}| = |B_{i3}| = 3$ and for the
    remaining agent $\{|B_{i1}|, |B_{i2}|, |B_{i3}|\} = \{2, 3, 4\}$.
\end{lemma}

\begin{proof}
    Exchange $\mathrm{MMS}_i$ for $x_i$ in the proofs of Propositions 20--28 in
    the full version of \citeauthor{Feige:2022}'s paper~\cite{Feige:2022}.
\end{proof}

\Cref{lem:feige-congregation} plays a key role in the proof of \cref{thr:c=6}.
The lemma allows us to use reductions to instances with three agents and nine
goods, despite the fact that MMS allocations do not necessarily exist for such
instances. After the reduction, it suffices for our use case to show that the
conditions of the lemma hold with $x_i$ set to the agent's MMS in the original
instance, rather than the possibly increased MMS in the three agent instance.

The following two lemmas, generalizing a couple of lemmas of \citet{Feige:2022},
allow us to simplify the proof of \cref{thr:c=6} to showing that all instances
with four agents and ten goods have MMS allocations.

\begin{lemma}\label{lem:mostly-overlapping-2}
    Let $I = \langle N, M, V \rangle$ be an ordered instance. If there exists a
    good $g \in M$ such that at most one agent does not have an MMS partition
    in which the bundle containing $g$ is of cardinality two, then there exists
    a valid reduction consisting of a single agent and a bundle containing $g$
    and one other good.
\end{lemma}

\begin{proof}
    One of the bundles of cardinality two containing $g$ is dominated by the
    others. If all agents have an MMS partition where the bundle containing $g$
    is of cardinality two, then allocating the dominated bundle to an agent that
    values it at MMS or higher is a valid reduction. Otherwise, let $i$ be the
    remaining agent and $B$ the dominated bundle. If $v_i(B) \ge \mu_i$,
    allocating $B$ to $i$ is a valid reduction. Otherwise, by
    \cref{lem:2-item-reduction}, allocating $B$ to an agent that values it at
    MMS or higher is a valid reduction.
\end{proof}

\begin{lemma}\label{lem:2n+2}
    Let $I = \langle N, M, V \rangle$ be an ordered instance with $n$ agents and
    $m \le 2n + 2$ goods, then there exists a valid reduction removing a single
    agent and a bundle containing one or two goods.
\end{lemma}

\begin{proof}
    If there is an agent $i \in N$ with $v_{i1} \ge \mu_i$, then allocating
    $\{1\}$ to $i$ is a valid reduction. If $v_{i1} < \mu_i$ for all agents $i
    \in N$, then by the assumption that $m \le 2n + 2$ each agent has at least
    $n - 2$ bundles in their MMS partition of cardinality 2. If for some $i \in
    N$, a bundle $B$ of cardinality 2 in her MMS partition does not contain a
    good in $\{1, 2, \dots, n - 1\}$, then $v_i(\{n, n + 1\}) \ge v_i(B) \ge
    \mu_i$ and allocating $\{n, n + 1\}$ to $i$ is a valid reduction by
    \cref{lem:pigeonhole-reduction-2}. Otherwise, there are $n(n - 2)$ bundles
    of cardinality $2$ in the MMS partitions, each intersecting with $\{1, 2,
    \dots, n - 1\}$. Thus, there is at least one $g \in \{1, 2, \dots, n - 1\}$
    that appears in at least $n - 1$ bundles of cardinality $2$ and a valid
    reduction with a single agent and two goods exists by
    \cref{lem:mostly-overlapping-2}.
\end{proof}

Note that \cref{lem:2n+2} provides a simplified proof for the existence of MMS
allocations for all instances with $n$ agents and $m \le n + 5$ goods. The lemma
shows that for any $n > 2$ we can repeatedly reduce the instance until there are
$n' = 2$ agents and $m \le n' + 5$ goods. Since all instances with two agents
have MMS allocations, it follows that an MMS allocation exists for all $n$ and
$m \le n + 5$.

We are now ready to prove \cref{thr:c=6}.

\repeattheorem{\ref{thr:c=6}}{\csixtheoremcontent}

\begin{proof}
    By \cref{lem:2n+2,lem:existance-feige} it suffice to show that all instances
    with $4$ agents and $10$ goods have an MMS allocation. Assume for the rest
    of the proof that we have an ordered instance $I = \langle N = \{1, 2, 3,
    4\}, M = \{1, 2, \dots, 10\}, V \rangle$.

    By \cref{lem:2n+2}, if $v_{i1} < \mu_i$ for each $i \in N$, there exists a
    valid reduction to an instance with $n = 3$ and $m = 8$, for which MMS
    allocations always exist. Further, by \cref{lem:2-item-reduction-from-1} if
    $v_{ig} \ge \mu_i$ for some $i \in N$ and $g \in M$, then there is either
    $i' \in (N \setminus \{i\})$ with $v_{i'}(g) \ge \mu_{i'}$ or a valid
    reduction exists to an instance with $n = 3$ and $m = 8$. Consequently, if
    $v_{i2} \ge \mu_i$, there is either a valid reduction to an instance with $n
    = 3$ and $m = 8$ by \cref{lem:2-item-reduction-from-1} or by allocating
    $\{1\}$ to $i$ and $\{2\}$ to $i'$, to an instance with $n = 2$ and $m = 8$.

    For an MMS allocation to not exist, there must be at least two agents that
    value good 1 at MMS or higher and no agent that values good 2 at MMS or
    higher. If this is the case, then there are two agents with MMS partitions
    restricted to the following types: $(1, 2, 2, 5)$, $(1, 2, 3, 4)$ and $(1,
    3, 3, 3)$. The remaining agents may either also have MMS partitions of the
    preceding types or of types $(2, 2, 2, 4)$ and $(2, 2, 3, 3)$. We wish to
    show that on a case-by-case basis, based on the types of MMS partitions
    present, selecting a specific agent to allocate $\{1\}$ to allows us to use
    \cref{lem:feige-congregation} to show there is a way to provide the
    remaining agents with at least their MMS in $I$. Note that since allocating
    $\{1\}$ to an agent $i \in N$ with $v_{i1} \ge \mu_i$ is a valid reduction,
    the MMS of an agent $i'$ in the three agent instance is at least as high as
    her MMS in the four agent instance. Thus, using $\mu_{i'}^I$ for $x_{i'}$ is
    valid.

    \paragraph{Any partition of type $\mathbf{(1, 2, 2, 5)}$.} Let $i$ be an
    agent with an MMS partition of type $(1, 2, 2, 5)$. Allocating $\{1\}$ to
    any other agent that values $\{1\}$ at MMS or higher, would mean that there
    is a 3-partition of the remaining goods of type $(2, 2, 5)$ where each
    bundle is valued at no less than $\mu_i^I$ by $i$.

    \paragraph{Any partition of type $\mathbf{(2, 2, 2, 4)}$.} Let $i$ be an agent
    with an MMS partition of type $(2, 2, 2, 4)$. Then, good $1$ is either in a
    bundle of size 2 or the bundle of size 4. In either case, after allocating
    $\{1\}$ to some other agent, we can merge the bundle in the MMS partition
    that contained $1$ with another bundle to create a $3$-partition of type
    $(2, 2, 5)$, where all bundles are worth at least $\mu_i^I$ to $i$.

    \paragraph{Any partition of type $\mathbf{(2, 2, 3, 3)}$.} Let $i$ be an agent
    with an MMS partition of type $(2, 2, 3, 3)$. After allocating $\{1\}$ to
    some other agent, the MMS partition of $i$, with good $1$ removed, can be
    modified in two different ways to create $3$-partitions of different types
    (some subset of two types from $(2, 3, 4)$, $(3, 3, 3)$ and $(2, 2, 5)$),
    where each bundle is valued at no less than $\mu_i^I$ by $i$. The two
    different types can be obtained by simply merging the bundle that contained
    $1$ with either a bundle of size 2 or 3.

    \paragraph{Only partitions of types $\mathbf{(1, 2, 3, 4)}$ and $\mathbf{(1,
    3, 3, 3)}$.} Partitions of type $(1, 2, 3, 4)$ turn into partitions of type
    $(2, 3, 4)$ and partitions of type $(1, 3, 3, 3)$ turn into partitions of
    type $(3, 3, 3)$. Since there are four agents, it is always possible to
    choose the agent to give $\{1\}$ to such that there remains either three
    agents with a partition of type $(3, 3, 3)$ or at least two agents with a
    partition of type $(2, 3, 4)$.
\end{proof}

\section{Proof of \cref{thr:c=7}}

\repeattheorem{\ref{thr:c=7}}{\cseventheoremcontent}

\begin{proof}
    Let $I = \langle N, M, V \rangle$ be an ordered instance with $n \ge 8$
    agents and $m = n + 7$ goods. If $n > 8$, then by \cref{lem:2n+2} there
    exists either a reduction to an instance with $n' \ge 8$ agents and $m' \le
    n' + 6$ goods, for which an MMS allocation always exists, or after repeated
    applications a reduction to an instance with $8$ agents and $15$ goods. For
    any instance with $8$ agents and $15$ goods, if there is an agent $i$ with
    $v_{i3} < \mu_i$, then $\mu_i > 0$ and any MMS partition of $i$ contains $5$
    bundles of size 2 and one bundle of size $3$. Thus, an MMS allocation exists
    by \cref{lem:one-3+-good-bundle}. We now assume that $v_{i3} \ge \mu_i$ for
    all $i \in N$ and that there are $8$ agents and $15$ goods in $I$.

    If there is an agent $i \in N$ with $\mu_i = 0$ or both $\mu_i > 0$ and
    $v_i(\{8, 9\}) \ge \mu_i$, then an MMS allocation exists, as allocating
    $\{8, 9\}$ to $i$ is a valid reduction by \cref{lem:pigeonhole-reduction-2}
    and by \cref{thr:c=6} an MMS allocation exists for any instance with $7$
    agents and $13$ goods. Thus, assume that $\mu_i > 0$ and $v_i(\{8, 9\}) <
    \mu_i$ for all $i \in N$.

    If $v_{i6} \ge \mu_i$ for some $i \in N$, then either
    \begin{enumerate}
        \item $v_{i'}(6) < \mu_{i'}$ for all $i' \in N \setminus \{i\}$,
        \item there is $i' \in N \setminus \{i\}$ with $v_{i'}(5) \ge \mu_{i'}$
            and $v_{i''}(5) < \mu_{i''}$ for all $i'' \in N \setminus \{i,
            i'\}$; or
        \item there are distinct $i',i'' \in N \setminus \{i\}$ with $v_{i'}(5)
            \ge \mu_{i'}$ and $v_{i''}(4) \ge \mu_{i''}$.
    \end{enumerate}

    \noindent
    In case 1., allocating $\{6, 15\}$ to $i$ is by
    \cref{lem:2-item-reduction-from-1} a valid reduction to an instance with $7$
    agents and $13$ goods for which an MMS allocation exists.  In case 2.,
    allocating $\{6, 14\}$ to $i$ and $\{5, 15\}$ to $i'$ is a valid reduction
    to an instance with $6$ agents and $11$ goods for which an MMS allocation
    exists. Finally, in case 3., allocating $\{4\}$ to $i''$, $\{5\}$ to $i'$,
    $\{6\}$ to $i$ and to three other agents each a good from $\{1, 2, 3\}$, is
    a valid reduction to an instance with two agents. All instances with two
    agents have MMS allocations (\cref{lem:mms-existence-few-agents}).  Thus, we
    assume that $v_{i6} < \mu_i$ for all $i \in N$.

    If there is an agent $i \in N$ with an MMS partition in which at most one
    bundle contains more than 2 goods, an MMS allocation exists by
    \cref{lem:one-3+-good-bundle}. Similarly, if there is an agent $i \in N$
    with an MMS partition in which there is a bundle with 6 or more goods, then,
    since $v_{i6} < \mu_i$, there is only a single bundle in the MMS partition
    that contains more than two goods and and MMS allocation exists.

    We wish to show that depending on the instance $I$, there either exists a
    valid reduction to an instance we know an MMS allocation exists for, or we
    can construct an MMS allocation directly. Due to the earlier
    assumptions about the instance $I$, \cref{lem:mms-partition-layout}
    guarantees that each agent has an MMS partition of one of five types:
    \begin{itemize}
        \item $t_1 = (1, 1, 1, 1, 1, 2, 3, 5)$
        \item $t_2 = (1, 1, 1, 1, 1, 2, 4, 4)$
        \item $t_3 = (1, 1, 1, 1, 2, 2, 3, 4)$
        \item $t_4 = (1, 1, 1, 1, 2, 3, 3, 3)$
        \item $t_5 = (1, 1, 1, 2, 2, 2, 3, 3)$.
    \end{itemize}
    Specifically, we know that in an MMS partition with $k$ bundles of size 1,
    these are $\{1\}, \{2\}, \dots, \{k\}$. Furthermore, since $v_i(\{8, 9\})
    \le \mu_i$ and $v_{i3} \le \mu_i$ for all $i \in N$, any bundle of size two
    contains at least one good in $\{4, 5, 6, 7\}$.

    We proceed on a case-by-case basis, based on combinations of types of MMS
    partition of the agents. For simplicity, we say that an agent has type
    $t_j$ if the agent has an MMS partition of the given type for $I$. An agent
    may have multiple types. However, we assume that an agent is given the type
    $t_j$ with the lowest possible value of $j$ for which the agent has an MMS
    partition. Thus, if an agent $i \in N$ has type $t_j$ with $j \ge 3$, then
    $v_{i5} < \mu_i$ and if $j = 5$, then $v_{i4} < \mu_i$.

    We start by considering cases in which there is at least one agent with type
    in $\{t_1, t_2\}$.

    \paragraph{Between one and four agents with type in $\mathbf{\{t_1,
    t_2\}}$.} Let $k$ be the number of agents with type in $\{t_1, t_2\}$.
    If $k > 1$ allocate the bundles $\{1\}, \dots, \{k - 1\}$ to $k - 1$ of the
    agents with type in $\{t_1, t_2\}$. This is a valid reduction by
    \cref{lem:single-item-reduction}. Since there is now only one agent of type
    in $\{t_1,t_2\}$, allocating $\{5, 15\}$ extends the valid reduction
    due to \cref{lem:2-item-reduction-from-1} and we are left with an instance
    with $n - k \ge 4$ agents and $(n - k) + 6$ goods, for which an MMS
    allocation always exists.

    \paragraph{At least five agents with type in $\mathbf{\{t_1, t_2\}}$.} Let
    $N'$ be a set of five agents with type in $\{t_1, t_2\}$ and allocate
    $\{1\}$, $\{2\}$ and $\{3\}$ to the three agents in $N \setminus N'$, along
    with $\{4\}$ and $\{5\}$ to two arbitrary agents in $N'$. Then we have three
    agents left, where the bundles remaining in the MMS partition of any of the
    agents form a $3$-partition of $M \setminus \{1, 2, 3, 4, 5\}$ of type $(2,
    3, 5)$ or $(2, 4, 4)$, where the value of each bundle is at least MMS to the
    agent. Since each bundle of size $2$ contains $g \in \{6, 7\}$, the
    reduction can through \cref{lem:mostly-overlapping-2} be extended by
    allocating a bundle of size two to one of the remaining agents. Thus, a
    reduction to an instance with two agents exists and an MMS allocation must
    exist.
    \\

    \noindent
    Since any instance with at least one agent of type $t_1$ or $t_2$ has an MMS
    allocation, we can now consider instances where the agents only have types
    $t_3$, $t_4$ and $t_5$.

    \paragraph{At least seven agents with type $\mathbf{t_4}$.} If there are
    at least seven agents of type $t_4$, then there is $g \in \{5, 6, 7\}$ that
    appears in the bundle of size two in the MMS partition of at least three of
    them. Of these three bundles, one, $B$, is dominated by the others. Let
    $N'$ denote a set of three such agents. Let $N'' = N' \cup \{i, i'\}$, where
    $i$ and $i'$ are distinct agents of type $t_4$ in $N \setminus N'$. If
    $v_i(B) < \mu_i$ or $v_{i'}(B) < \mu_{i'}$, then there exits a valid
    reduction to an instance with $4$ agents and $10$ goods by doing the
    following, assuming w.l.o.g.\ that $v_i(B) < \mu_i$. First, allocate $\{1\}$,
    $\{2\}$ and $\{3\}$ to the agents in $N \setminus N''$. This is a valid reduction
    to an instance with $5$ agents and $12$ goods. Further, if $v_{i'}(B) \ge
    \mu_i$, then the valid reduction can be extended by allocating $B$ to $i'$,
    as this does not decrease the MMS of the other agents in $N''$ due to either
    \cref{lem:2-item-reduction} or \ref{lem:domination-reduction}. If
    $v_{i'}(B) < \mu_i$, then allocating $B$ to the agent $i'' \in N'$ who's MMS
    partition $B$ came from extends the valid reduction by the same argument.
    Since there in these cases exists a reduction to an instance with four
    agents and ten goods, an MMS allocation exists.

    If, on the other hand, $v_i(B) \ge \mu_i$ and $v_{i'}(B) \ge \mu_{i'}$, then
    we will show that an MMS allocation exists in which each agent in $N
    \setminus N''$ receives a good in $\{1, 2, 3\}$, $i$ receives $B$ and
    $\{4\}$ is given to $i'$. We thus need to show that there is a way to
    allocate the goods in $M \setminus (\{1, 2, 3, 4\} \cup B)$ to the three
    agents in $N'$ such that each receives a bundle worth at least her MMS in
    $I$. Note that for any agent $i'' \in N'$, $i''$ has an MMS partition $A$ of
    type $(1, 1, 1, 1, 2, 3, 3, 3)$ where the bundle of size $2$ contains $g$.
    After removing the goods in $\{1, 2, 3, 4\}$ from $A$, we are left with a
    $4$-partition of $M \setminus \{1, 2, 3, 4\}$ in which each bundle is worth
    at least $\mu_{i''}$ to $i''$. The $4$-partition has type $(2, 3, 3, 3)$ and
    the bundle $B'$ of size 2 contains $g$.  If $B' \neq B$, then swap the
    position of the good $g' \in B' \setminus \{g\}$ for the position of the
    good $g'' \in B \setminus \{g\}$. Since $g' < g''$, the the $4$-partition
    now contains $B$ and three bundles of size $3$. Each of the bundles of size
    $3$ has a value of at least $\mu_{i''}$ to $i''$. Removing $B$ produces a
    $3$-partition of $M \setminus (\{1, 2, 3, 4\} \cup B)$ of type $(3, 3, 3)$
    such that each bundle is worth at least $\mu_{i''}$ to $i''$. Since each
    agent in $N'$ has such a $3$-partition, an MMS allocation exists by
    \cref{lem:feige-congregation}.

    \paragraph{Less than seven agents with type $\mathbf{t_4}$.} Each agent of
    type $t_3$ and $t_5$ has at least two bundles of size two in their MMS
    partition that each overlaps with $\{5, 6, 7\}$. Each agent of type $t_4$
    has one bundle of this kind in their MMS partition. Hence, with less than
    seven agents of type $t_4$, there are at least $4 + 6 = 10$ bundles of size
    two in the MMS partitions that overlap with $\{5, 6, 7\}$. There is $g \in
    \{5, 6, 7\}$ such that at least $\lceil 10/3 \rceil = 4$ of these contain
    $g$. Thus, there is a subset of five agents of which at least $4$ have a
    bundle of size 2 containing $g$ in their MMS partition. Allocating $\{1\}$,
    $\{2\}$ and $\{3\}$ to the other agents means that
    \cref{lem:mostly-overlapping-2} guarantees a further extension of the
    reduction to an instance with $4$ agents and $10$ goods, for which an MMS
    allocation always exists.
\end{proof}

\section{Proof of \Cref{lem:one-3+-good-bundle}}

\begin{proof}[Proof of \cref{lem:one-3+-good-bundle}]
    If $\mu_i = 0$, then allocating $\{n, n + 1\}$ to $i$ is by
    \cref{lem:pigeonhole-reduction-2} a valid reduction to an instance with $(n
    - 1) \ge n_{c - 1}$ agents and $n + c - 2 = (n - 1) + (c - 1)$ goods, for
    which an MMS allocation exists.

    Now assume that $\mu_i > 0$. Consequently, each bundle in $A$ contains at
    least one good. If there is a perfect matching between the agents in $N$ and
    bundles in $A$ valued at MMS or higher, then this matching is an MMS
    allocation. Assume that such a perfect matching does not exist and let $A' =
    \{B \in A : |B| \le 2\}$. Then $|A'| \ge n - 1$ and there exists no perfect
    matching between agents in $N \setminus \{i\}$ and bundles they value at MMS
    or higher in $A'$. If a perfect matching of this kind existed, then a
    perfect matching between $N$ and bundles in $A$ would also exist, as $v_i(B)
    \ge \mu_i$ for all $B \in A$. Thus, by Hall's marriage theorem, there exists
    $N' \subseteq (N \setminus \{i\})$ and $A'' \subset A'$ such that $|N'| >
    |A''|$ and $v_{i'}(B) < \mu_{i'}$ for $i'
    \in N', B \in (A' \setminus A'')$. In other words, no agent in $N'$ values
    any bundle in $A' \setminus A''$ at MMS or higher. Additionally, by
    \cref{lem:goods-at-mms}, $|A' \setminus A''| \le \min(n - 1, c)$. We have
    that $|N \setminus N'| \le |A' \setminus A''|$ and since $A$ is an MMS
    partition of $i$, with $i \in N \setminus N'$, \cref{thr:envy-free-matching}
    guarantees that there exists a non-empty envy-free matching $M$ with regards
    to the agents in the graph consisting of the agents in $N \setminus N'$ and
    bundles in $A' \setminus A''$, with edges between agent-bundle pairs if the
    agent values the bundle at MMS or higher. We wish to show that this matching
    can be converted into a valid reduction of $x$ agents and $2x$ goods, where
    $x$ is the number of agents in the envy-free matching. The valid reduction
    will be constructed in the following way:
    \begin{enumerate}
        \item Allocate all bundles in the matching containing two goods to their
            matched agent.
        \item For each bundle in the matching containing one good, allocated it
            to the matched agent along with the worst remaining unmatched good.
    \end{enumerate}
    \noindent
    As each agent that receives a bundle values it at MMS or higher, we need to
    show that for any unmatched agent $i'$, their MMS has not decreased.  For
    $i'$, consider each allocation in step 1 and step 2 as individual reductions
    performed in turn in smaller and smaller instances. Then, the only way that
    $i'$'s MMS can decrease is if it decreases after allocating one of the
    bundles. However, by \cref{lem:2-item-reduction,lem:2-item-reduction-from-1}
    this can only occur if the matched bundle (of two goods in step 1 and one
    good in step 2) is valued at MMS or higher. Since $i'$'s value of a bundle
    does not change, and any matched bundle is valued at less than MMS by $i'$
    before the step, the only way for $i'$'s MMS to decrease is if it has
    already decreased, which is a contradiction. Thus, $i'$'s MMS does not
    decrease and we have a valid reduction.

    The valid reduction removes $x$ agents and $2x$ goods. Thus, the reduced
    instance has $n - x$ agents and $(n - x) + (c - x)$ goods. If $c - x \le 5$,
    then an MMS allocation always exists. Otherwise, since $n_{c'} > n_{c' - 1}$
    for $c > c' > 6$, we have $n - x \ge n_{c - x}$ and an MMS allocation
    exists.
\end{proof}

\end{document}